\documentclass{article}
\title{\textbf{It'll probably work out: improved list-decoding through random operations}\footnote{AR's research supported in part by NSF CAREER grant CCF-0844796 and NSF grant CCF-1161196. MW's research supported in part by a Rackham predoctoral fellowship.}}

\newcommand{\rate}[1]{R^*(#1)}

\author{\textsc{Atri Rudra}\footnotemark[2] \and \textsc{Mary Wootters}\footnotemark[3]}
\date{\today\\
\vspace*{4mm}
\footnotemark[2]~~Department of Computer Science and Engineering,\\
University at Buffalo, SUNY\\
{\tt atri@buffalo.edu}\\
\vspace*{2mm}
\footnotemark[3]~~Department of Mathematics,\\
University of Michigan\\
{\tt wootters@umich.edu}}
\usepackage{amssymb,amsmath,amsthm,amsfonts}
\usepackage[ruled,vlined]{algorithm2e}
\usepackage{framed,fullpage,nicefrac,graphicx,comment}
\usepackage{tikz}
\usetikzlibrary{shapes}
\usetikzlibrary{shapes.callouts}
\usetikzlibrary{decorations,positioning}
\usetikzlibrary{arrows,decorations.pathmorphing,decorations.shapes}
\usepackage{cite}

\newcommand{\pl}{\ensuremath{\,\mathrm{\mathbf{pl}}}}

\newcommand{\A}{\ensuremath{\mathcal{A}}}
\newcommand{\B}{\ensuremath{\mathcal{B}}}

\newcommand{\cC}{\ensuremath{\mathcal{C}}}
\newcommand{\cD}{\ensuremath{\mathcal{D}}}
\newcommand{\cE}{\ensuremath{\mathcal{E}}}

\newcommand{\cU}{\ensuremath{\mathcal{U}}}

\newcommand{\R}{{\mathbb R}}

\newcommand{\F}{{\mathbb F}}

\newcommand{\PR}[1]{{\mathbb{P}}\left\{ #1\right\}}
\newcommand{\E}[1]{\mathbb{E}\left[ #1\right]}
\newcommand{\EE}{\mathbb{E}}

\newcommand{\norm}[1]{\left\|#1\right\|}
\newcommand{\twonorm}[1]{\left\|#1\right\|_2}

\newcommand{\decnorm}[2]{\left\|#1\right\|_{#2}}

\newcommand{\infnorm}[1]{\decnorm{#1}{\infty}}

\newcommand{\inabs}[1]{\left|#1\right|}

\newcommand{\ind}[1]{\ensuremath{\mathbf{1}_{#1}}}

\newcommand{\supp}{\mathrm{Supp}}
\newcommand{\agr}{\mathrm{agr}}

\newcommand{\ip}[2]{\ensuremath{\left\langle #1,#2\right\rangle}}
\newcommand{\inset}[1]{\left\{#1\right\}}
\newcommand{\inparen}[1]{\left(#1\right)}

\newcommand{\inbrac}[1]{\left\{#1\right\}}
\newcommand{\suchthat}{\,:\,}

\newcommand{\eps}{\varepsilon}

\newcommand{\mkw}[1]{\textcolor{red}{(#1 --mary)}}
\newcommand{\ar}[1]{\textcolor{blue}{(#1 --Atri)}}

\newtheorem{theorem}{Theorem} 

\newtheorem{lemma}{Lemma} 
\newtheorem{definition}{Definition}

\newtheorem{conjecture}[theorem]{Conjecture} 
\newtheorem{cor}{Corollary} 

\newtheorem{remark}{Remark}
\newtheorem{ques}[remark]{Question}
\newtheorem{claim}[theorem]{Claim}
\newtheorem{proposition}{Proposition}

\newtheorem{problem}[theorem]{Problem}

\newcommand{\dprod}[1]{\cU_{{#1},\mathrm{dp}}}

\begin{document}
\maketitle
\thispagestyle{empty}

\begin{abstract}
In this work, we introduce a framework to study the effect of random operations on the combinatorial list decodability of a code.  
The operations we consider correspond to row and column operations on the matrix obtained from the code by stacking the codewords together as columns.  This captures many natural transformations on codes, such as puncturing, folding, and taking subcodes; we show that many such operations can improve the list-decoding properties of a code.  There are two main points to this.  First, our goal is to advance our (combinatorial) understanding of list-decodability, by understanding what structure (or lack thereof) is necessary to obtain it.  Second, we use our more general results to obtain a few interesting corollaries for list decoding:
\begin{enumerate}
\item We show the existence of binary codes that are combinatorially list-decodable from $1/2-\eps$ fraction of errors with optimal rate $\Omega(\eps^2)$ that can be encoded in {\em linear} time.
\item We show that any code with $\Omega(1)$ relative distance, when randomly folded, is combinatorially list-decodable $1-\eps$ fraction of errors with high probability. This formalizes the intuition for why the folding operation has been successful in obtaining codes with optimal list decoding parameters; previously, all arguments used algebraic methods and worked only with specific codes.
\item We show that any code which is list-decodable with suboptimal list sizes has many subcodes which have near-optimal list sizes, while retaining the error correcting capabilities of the original code.  This generalizes recent results where subspace evasive sets have been used to reduce list sizes of codes that achieve list decoding capacity.
\end{enumerate}
The first two results follow from the techniques of Wootters (STOC 2013) and Rudra and Wootters (STOC 2014); one of the main technical contributions of this paper is to demonstrate the generality of the techniques in those earlier works.  The last result follows from a simple direct argument. 
\end{abstract}

\newpage
\section{Introduction}
The goal of \em error correcting codes \em is to enable communication between a sender and receiver over a noisy channel.
For this work, we will think of a code $\cC$ of block length $n$ and size $N$ over an alphabet $\Sigma$ as an $n \times N$ matrix over $\Sigma$, 
where each column in the matrix $\cC$ is called a {\em codeword}. 
The sender and receiver can use $\cC$ for communication as follows. 
Given one of $N$ messages---which we think of as indexing the columns of $\cC$---the sender transmits the corresponding codeword over a noisy channel.  
The receiver gets a corrupted version of the transmitted codeword and aims to recover the originally transmitted codeword (and hence the original message).  Two primary quantities of interest are the fraction $\rho$ of errors that the receiver can correct (the \em error rate\em); and the redundancy of the communication, as measured by the \em rate \em $R: = \frac{\log_{|\Sigma|}{N}}{n}$ of the code. 
The central goal is to design codes $\cC$ so that both $R$ and $\rho$ are large.

A common approach to this goal is to first design a code matrix $\cC_0$ that is ``somewhat good," and to modify it to obtain a better code $\cC$.  Many of these modifications correspond to row or column operations on the matrix $\cC_0$: for example, dropping of rows or columns, taking linear combinations of rows or columns, and combining rows or columns into ``mega" rows or columns.  In this work, we study the effects of such row- and column-operations on the {\em list decodability} of the code $\cC_0$.  

\paragraph{List decoding.}
In the list decoding problem~\cite{elias,wozencraft}, the receiver is allowed to output a small list of codewords that includes the transmitted codeword, instead of having to pin down the transmitted codeword exactly.  The remarkable fact about list decoding is that the receiver may correct twice as many adversarial errors as is possible in the unique decoding problem.  Exploiting this fact has led to many applications of list decoding in complexity theory and in particular, pseudorandomness.\footnote{See the survey by Sudan~\cite{madhu-survey} and
Guruswami's thesis~\cite{venkat-thesis} for more on these applications.} 

Perhaps the ultimate goal of list decoding  research is to solve the following problem.
\begin{problem}
\label{prob:ideal}
For $\rho \in (0,1 - 1/q)$, construct codes with rate $1 - H_q(\rho)$ that can 
 correct $\rho$ fraction of errors with linear time encoding and linear time decoding.\footnote{One needs to be careful about the machine model when one wants to claim linear runtime. In this paper we consider the RAM model. For the purposes of this paper, it is fine to consider linear time to mean linear number of $\F_q$ operations and the alphabet size to be be small, say polynomial in $1/\eps$.}
Above, $H_q$ denotes the $q$-ary entropy, and $1 - H_q(\rho)$ is known to be the optimal rate.
\end{problem}

Even though much progress has been made in algorithmic list decoding, we are far from answering the problem above in its full generality. If we are happy with polynomial time encoding and decoding (and large enough alphabet size), then the problem was solved by Guruswami and Rudra~\cite{GR08}, and improved by several follow-up results~\cite{GW13, K12, GX12, GX13, DL12,GK13}.
However, even with all of this impressive work on algorithmic list decoding, the landscape of list-decoding remains largely unexplored.  
First, while the above results offer concrete approaches to Problem \ref{prob:ideal}, we do not have a good characterization of which codes are even \em combinatorially \em list-decodable at near-optimal rate.  Second, while we have polynomial-time encoding and decoding, linear-time remains an open problem.  In this work, we make some progress in both of these directions.

\paragraph{New codes from old: random operations.}
In this paper, we develop a framework to study the effect of random operations on the list-decodability of a code.  Specific instantiations of these operations are a common approach to Problem \ref{prob:ideal}.
For example, 
\begin{enumerate}
\item In the Folded Reed-Solomon codes mentioned above, one starts with a Reed-Solomon code and modifies it by applying a \em folding \em operation to each codeword. In the matrix terminology, we  bunch up rows to construct ``mega" rows. 
\item In another example mentioned above~\cite{GX13}, one starts with a Reed-Solomon code and picks certain positions in the codeword, and also throws away many codewords---that is, one applies a \em puncturing \em operation the codewords, and then considers a \em subcode. \em  In matrix terminology, we drop rows and columns.
\item In~\cite{luca,dp-codes}, the \em direct product \em operation and the \em XOR \em operation are used to enhance the list-decodability of codes. In matrix terminology, the direct product corresponds to bunching rows and the XOR operation corresponds to taking inner products of rows. 
\item In~\cite{GI01,GI03,GI05}, the \em aggregation \em operation is used to construct efficiently list-decodable codes out of list-recoverable codes.  In matrix terminology, this aggregation again corresponds to bunching rows.
\end{enumerate}
However, in all of these cases, the operations used are very structured; in the final two, the rate of the code also takes a hit.\footnote{It must be noted that in the work of~\cite{luca,dp-codes} the main objective was to obtain sub-linear time list decoding and the suboptimal rate is not crucial for their intended applications.}  It is natural to ask how generally these operations can be applied. 
In particular, if we considered random versions of the operations above, can we achieve the optimal rate/error rate/list size trade-offs? 
If so, this provides more insight about why the structured versions work.

Recently the authors showed in~\cite{rw14} that the answer is ``yes" for puncturing of the rows of the code matrix: if one starts with \em any \em code with large enough distance and randomly punctures the code, then with high probability the resulting code is nearly optimally combinatorially list-decodable.  In this work, we extend those results to other operations.

\subsection{Our contributions and applications}
The contributions of this paper are two-fold.  First, the goal of this work is to improve our understanding of (combinatorial) list-decoding.  What is it about these structured operations that succeed?  How could we generalize?  Of course, this first point may seem a bit philosophical without some actual deliverables.  To that end, we show how to use our framework to address some open problems in list decoding.  We outline some applications of our results below.

In order to state our main results, we pause briefly to set the quantitative stage.
There are two main parameter regimes for list-decoding, and we will focus on both in this paper.
In the first regime, corresponding the the traditional communication scenario, the error rate $\rho$ is some constant $0 < \rho < 1 - 1/q$.
In the second regime, motivated by applications in complexity theory, the error rate $\rho$ is very large.   
For $q$-ary codes, these applications require correction from a $\rho = 1-1/q-\eps$ fraction of errors, for small $\eps>0$. 
In both settings, the best possible rate is given by
\[ R^* = 1 - H_q(\rho), \]
where $H_q$ denotes the $q$-ary entropy.  In the second, large-$q$, regime, we may expand $H_q(1 - 1/q - \eps)$ to obtain an expression
\[\rate{q,\eps} := 1 - H_q(1 - 1/q - \eps) = \min\inset{ \eps, \frac{ q\eps^2}{2\log(q)} + O_q(\eps^3)}.\]
For complexity applications it is often enough to design a code with rate $\Omega(\rate{q,\eps})$
with the same error correction capability.

\subsubsection{Linear time encoding with near optimal rate.} We first consider the special case of Problem~\ref{prob:ideal} that concentrates on the encoding complexity for binary codes in the high error regime:
\begin{ques}
\label{ques:enc}
Do there exist binary codes with rate $\Omega(\eps^2)$ that can be encoded in linear time and are (combinatorially) list-decodable from a $1/2-\eps$ fraction of errors?
\end{ques}

Despite much progress on related questions, 
obtaining linear time encoding with (near-)optimal rate is still open.
 More precisely, for $q$-ary codes (for $q$ sufficiently large, depending on $\eps$), Guruswami and Indyk showed that linear time encoding and decoding with near-optimal rate is possible for {\em unique} decoding~\cite{GI05}. For list decoding, they prove a similar result for list decoding but the rate is exponentially small in $1/\eps$~\cite{GI03}. This result can be used with code concatenation to give  a similar result for binary codes (see Appendix~\ref{app:binary-lin} for more details) but also suffers from an exponentially small rate.
If we allow for super-linear time encoding in Question \ref{ques:enc}, then it is known that the answer is yes.  Indeed, random linear codes will do the trick~\cite{ZP,cgv2012,woot2013} and have quadratic encoding time; 
In fact, near-linear time encoding with optimal rate also follows from known results.\footnote{For example, Guruswami and Rudra~\cite{GR10} showed that folded Reed-Solomon codes---which can be encoded in near-linear time---concatenated with random inner codes with at most logarithmic block length achieve the optimal rate and fraction of correctable errors tradeoff.}

\paragraph{Our results.}
We answer Question~\ref{ques:enc} in the affirmative.  To do this, we consider the row-operation on codes given by taking random XORs of the rows of $\cC_0$.
We show that this operation yields codes with rate $\Omega(\eps^2)$ that are combinatorially list-decodable from $1/2 - \eps$-fraction of errors, provided the original code has constant distance and rate.  Instantiating this by taking $\cC_0$ to be Spielman's code~\cite{spielman}, we obtain a linear-time encodable binary code which is nearly-optimally list-decodable.

\subsubsection{The folding operation, and random $t$-wise direct product.} The
result of Guruswami and Rudra~\cite{GR08} showed that when the {\em folding}
operation is applied to Reed-Solomon codes, then the resulting codes (called
folded Reed-Solomon codes) can be list decoded in polynomial time with optimal
rate. The folding operation is defined as follows.  
We start with a $q$-ary code $\cC_0$ of length $n_0$, and a partition of $[n_0]$ into $n_0/t$ sets of size $t$,
and we will end up with a $q^t$-ary code $\cC$ of length $n = n_0/t$. 
Given a codeword $c_0 \in \cC_0$, we form a new codeword $c \in \cC$ by ``bunching" together the symbols in each partition set and treating them as a single symbol.  A formal definition is given in Section \ref{sec:setup}. 
For large enough $t$, this
results in codes that can list decode from $1-\eps$ fraction of errors with
optimal rate~\cite{GR08,GX12,GX14} when one starts with Reed-Solomon or more
generally certain algebraic-geometric codes. In these cases, the partition for
folding is very simple: just consider $t$ consecutive symbols to form the $n/t$
partition sets.

Folding is a special case of $t$-wise aggregation of symbols.  Given a code $\cC_0$ of length $n_0$, we may form a new code $\cC_0$ of length $n$ by choosing $n$ subsets $S_1,\ldots, S_n \subset [n_0]$ and aggregating symbols according to these sets. This operation has also been used to good effect in the list-decoding literature: in~\cite{GI01,GI03,GI05}, the sets $S_i$ are defined using expander codes, and the original code $\cC_0$ is chosen to be \em list-recoverable. \em  This results in efficiently list-decodable codes, although not of optimal rate.
We can also view this $t$-wise aggregation as a puncturing of a $t$-wise direct product (where $n = {n_0 \choose t}$ and all sets of size $t$ are included).

There is a natural intuition for the effectiveness of the folding operation in~\cite{GR08,GR09}, and for the $t$-wise aggregation of symbols in~\cite{GI01,GI03,GI05}.
In short, making the symbols larger increases the size of the ``smallest corruptable unit," which in turn decreases the number of error patterns we have to worry about.  (See Section \ref{sec:fold} for more on this intuition). 
In some sense, this intuition is the reason that random codes over large alphabets can tolerate more error than random codes over small alphabets: indeed, an inspection of the proof that random codes obtain optimal list-decoding parameters shows that this is the crucial difference. 
Since
a random code over a large alphabet is in fact a folding of a random code over
a small alphabet, the story we told above is at work here.

Despite this nice-sounding intuition---which doesn't use anything specific about the code---the known results mentioned above do not use it, and rely crucially on specific properties of the original codes, and on algorithmic arguments.  It is natural to wonder if the intuition above can be made rigorous, and to hold for \em any \em original code $\cC_0$.
In particular,
\begin{ques}
\label{ques:fold}
Can the above intuition be made rigorous?  Precisely, 
are there constants $\delta_0,c_0 > 0$, so that for any $\eps > 0$, any code with distance at least $\delta_0$ and rate at most $c_0 \eps$ admits a $t$-wise folding (or other $t$-wise aggregation of symbols with $n = n_0/t$) for $t$ depending only on $\eps$, such that the resulting code is combinatorially list-decodable from a $1-\eps$ fraction of errors?  
\end{ques}
The first question mimics the parameters of folded Reed-Solomon codes; the second part is for the parameter regime of~\cite{GI01,GI03,GI05}.
Notice that both the requirements (distance $\Omega(1)$ and rate $O(\eps)$) are necessary.
Indeed, if the original code does not have distance bounded below by a constant, it is easy to come up with codes where the answer to the above question is ``no."  
The requirement of $O(\eps)$ on the rate of the original code is needed because folding preserves the rate, and the list-decoding capacity theorem implies that any code that can be list decoded from $1-\eps$ fraction of errors must have rate $O(\eps)$.

\paragraph{Our results.}  We answer Question \ref{ques:fold} in the affirmative by considering the operation of random $t$-wise aggregation.  We show that if $n = n_0/t$ (the parameter regime for $t$-wise folding), the resulting code is list-decodable from a $1-\eps$ fraction of errors, as long as $t = O(\log(1/\eps))$.  Our theory can also handle the case when $n \ll n_0$, and obtain near-optimal rate at the same time.

\subsubsection{Taking sub-codes.} 
The result of Guruswami and Rudra~\cite{GR08}, even though it achieves the optimal tradeoff between rate and fraction of correctable errors is quite far from achieving the best known combinatorial bounds on the worst-case list sizes. Starting with the work of Guruswami~\cite{G11}, there has been a flurry of work on using {\em subspace evasive subsets} to drive down the list size needed to achieve the optimal list decodability~\cite{GW13,DL12, GX12, GX13, GK13}. The basic idea in these works is the following: we first show that some code $\cC_0$ has optimal rate vs fraction of correctable tradeoff but with a large list size of $L_0$. In particular, this list lies in an affine subspace of roughly $\log{L_0}$ dimensions. A subspace evasive subset is a subset that has a small intersection with any low dimension subset. Thus, if we use such a subset to pick a subcode of $\cC_0$, then the resulting subcode will retain the good list decodable properties but now with smaller worst-case lists size. Perhaps the most dramatic application of this idea was used by Guruswami and Xing~\cite{GX13} who show that certain Reed-Solomon codes have (non-trivial) exponential list size and choosing an appropriate subcode with a subspace evasive subset reduces the list size to a constant.

However, the intuition that using a subcode can reduce the worst-case list size is not specifically tied to the algebraic properties of the code (i.e, to Reed-Solomon codes and subspace evasive sets). As above, it is natural to ask if this intuition holds more broadly.
\begin{ques}
\label{ques:subcode}
Given a code, does there always exist a subcode that has the same list decoding properties as the original code but with a smaller list size? In particular, is this true for random sub-codes?
\end{ques}

\paragraph{Our results.}
We answer Question~\ref{ques:subcode} by showing that for any code, a random subcode with the rate smaller only by an additive factor of $\eps$ can correct the same fraction of errors as the original code but with a list size of $O(1/\eps)$ as long as the original list size is at most $N^{\eps}$. 
Guruswami and Xing~\cite{GX13} showed that Reed-Solomon codes defined over (large enough) extension fields with evaluation points coming from a (small enough) subfield has non-trivial list size of $N^{\eps}$. Thus, our result then implies the random sub-codes of such Reed-Solomon codes are optimally list decodable.\footnote{Guruswami and Xing also prove a similar result (since a random subset can be shown to be subspace evasive) so ours gives an arguably simpler alternate proof.} 
We also complement this result by showing that the tradeoff between the loss in rate and the final list size is the best one can hope for in general. 
We also use the positive result to show another result: given that $\cC_0$ is optimally list decodable up to rate $\rho_0$, its random subcodes (with the appropriate rate) with high probability are also optimally list decodable for any error rate $\rho>\rho_0$.

\subsubsection{Techniques}
Broadly speaking, the operations we consider fall into two categories: row-operations and column-operations on the matrix $\cC$.  We use different approaches for the different types of operations.

For row operations (and Questions~\ref{ques:enc} and~\ref{ques:fold}) we use the machinery of~\cite{woot2013,rw14} in a more general context.  In those works, the main motivations were specific families of codes (random linear codes and Reed-Solomon codes).   In this work, we use the technical framework (implicit in) those earlier papers to answer new questions.  Indeed, one of the contributions of the current work is to point out that in fact these previous arguments apply very generally. 
For column operations, our results follow from a few simple direct arguments (although the construction for the lower bound requires a bit of care). 

\begin{remark}\label{rem:puncturing}
We will specifically handle all row operations on the code matrix mentioned at the beginning of the introduction. For column operations, we handle only column puncturing (taking random subcodes). For many operations, this is not actually an omission: some of the column-analogues of the row-operations we consider are redundant.  For example, taking random linear combinations of columns of a {\em linear code} has the same distribution as a random column puncturing. We do not handle bunching up of columns into mega columns, which would correspond to designing interleaved codes---see Section~\ref{sec:setup} for a formal definition---and we leave the solution of this problem as an open question.
\end{remark}

\subsection{Organization}
In Section~\ref{sec:setup}, we set up our formal framework and present an overview of our techniques in Section~\ref{sec:techniques}.
In Section~\ref{sec:gen}, we state and prove our results about the list-decodability of
codes under a few useful random operations; these serve to give examples for
our framework.  They also lay the groundwork for Section~\ref{sec:app}, where we return to the three applications we listed above, and
resolve Questions~\ref{ques:enc},~\ref{ques:fold}, and~\ref{ques:subcode}.  Finally, we conclude with some open questions.

\section{Set-up}
\label{sec:setup}
In this section, we set notation and definitions, and formalize our notion of row and column operations on codes.
Throughout, we will be interested in codes $\cC$ of length $n$ and size $N$ over an alphabet $\Sigma$.  
Traditionally, $\cC \subset \Sigma^n$ is a set of codewords.  As mentioned above, we will treat $\cC$ as a matrix in $\Sigma^{n \times N}$, with the codewords as columns.  We will abuse notation slightly by using $\cC$ to denote both the matrix and the set; which object we mean will be clear from context. For a prime power $q$, we will use $\F_q$ to denote the finite field with $q$ elements.

For $x,y \in \Sigma^n$, we will use $d(x,y)$ to denote the Hamming distance between $x$ and $y$, and we will use $\agr(x,y) := n - d(x,y)$ to denote the agreement between $x$ and $y$.  
We study the list-decodability of $\cC$: we say that $\cC$ is \em  $(\rho, L)$-{list-decodable} \em if for all $z \in \Sigma^n$, 
$|\inset{ c \in \cC \suchthat d(c,z) \leq \rho } | < L$.  In this work, we will also be interested in the slightly stronger notion of \em average-radius list-decodability. \em
\begin{definition}\label{def:avgrad}
A code $\cC \subset \Sigma^n$ is $(\rho, L)$-average-radius list-decodable if for all sets $\Lambda \subset \cC$ with $|\Lambda| = L$, 
\[ \max_z \sum_{c \in \Lambda} \agr( c, z) \leq (1-\rho) nL. \]
\end{definition}
Average-radius list-decodability implies list-decodability~\cite{gurnar2013,rw14}.  Indeed, the mandate of average-radius list decodability is that, for any $L$ codewords in $\cC$, they do not agree too much on average with their center, $z$.  On the other hand, standard list decodability requires that for any $L$ codewords in $\cC$, at least one does not agree too much with $z$.  As the average is always smaller than the maximum, standard list-decodability follows from average-radius list-decodability.

We will create new codes $\cC \in \Sigma^{n \times N}$ from original codes $\cC_0 \in \Sigma_0^{n_0 \times N_0}$; notice that we allow the alphabet to change, as well as the size and block length of the code.  We will consider code operations $f: \Sigma_0^{n_0 \times N_0} \to \Sigma^{n \times N}$ which act on rows and columns of the matrix $\cC_0$.  

We say that a \em basic row operation \em takes a code $\cC_0$ and produces a row of a new matrix $\cC$: that is, it is a function
\[r: \Sigma_0^{n_0\times N_0}\to \Sigma^{N_0}.\]
Two examples of basic row operations that we will consider in this paper are taking linear combinations of rows or aggregating rows.
That is:
\begin{itemize}
	\item[(a)] When $\Sigma = \Sigma_0 = \F_q$, and for a vector $v \in \F_q^{n_0}$, the row operation corresponding to linear combinations of rows is $r^{\text{(ip)}}_{v}: \F_q^{n_0 \times N} \to \F_q^{N}$,  given by
\[ r^{\text{(ip)}}_{v}(\cC_0) = v^T \cC_0. \]
	\item[(b)] Let $S \subset [n_0]$ be a set of size $t$, and let $\Sigma = \Sigma_0^t$.  Then the row operation corresponding to aggregating rows is $r^{\text{(agg)}}_{S}: \Sigma_0^{n_0 \times N} \to (\Sigma_0^t)^N$, given by
\[ r^{\text{(agg)}}_{S}(M) = \inparen{ \inparen{M_{i,1}}_{i \in S}, \inparen{ M_{i,2}}_{i \in S}, \ldots, \inparen{M_{i,N}}_{i \in S} }. \]
(Above, we have replaced $\cC_0$ with $M$ to ease the number of subscripts).
\end{itemize}
We will similarly consider \em basic column operations \em 
\[ c: \Sigma_0^{n_0 \times N_0} \to \Sigma^{n_0}, \]
which take a code $\cC_0$ and produce a new column of a matrix $\cC$.
Analogous to the row operations, we have the following two examples.
\begin{itemize}
	\item[(a)] When $\Sigma = \Sigma_0 = \F_q$, and for a vector $w \in \F_q^{N_0}$, we can consider
\[ c^{\text{(ip)}}_{w}(\cC_0) = \cC_0 w. \]
	\item[(b)] Let $T \subset [N_0]$ be a set of size $t$, and let $\Sigma = \Sigma_0^t$.  Then
\[ c^{\text{(agg)}}_T(M) = \inparen{ \inparen{M_{1,j}}_{j \in T}, \inparen{ M_{2,j}}_{j \in T}, \ldots, \inparen{M_{n, j}}_{j \in T} }. \]
\end{itemize}

The code operations that we will consider in this paper are distributions over a collection of random basic row operations or collection of random basic column operations:
\begin{definition}
\label{def:random-ops} A random row operation is a distribution $\cD$ over $n$-tuples of basic row operations.
We treat a draw $f = (r_1,\ldots,r_n)$ from $\cD$ as a code operation mapping $\cC_0$ to $\cC$ by defining the $i^{th}$ row of $\cC = f(\cC_0)$ to be $r_i(\cC_0)$.
Similarly, a random column operation is a distribution $\cD$ over $N$-tuples of basic column operations. 

We say a random row (column) operation $\cD$ has {\em independent symbols} ({\em independent codewords} resp.) if the coordinates are independent. We say a random row operation $\cD$ has symbols drawn {\em independently without replacement} if $(r_1,\dots,r_n)$ are drawn uniformly at random without replacement from some set $R$ of basic row operations.

Finally, for a random row operation $\cD$ and a sample $f$ from $\cD$ note that the columns of $f(\cC)$ are in one-to-one correspondence with the columns of $\cC$. Thus, we will overload notation and denote $f(c)$ for $c\in \cC$ to denote the column in $f(\cC)$ corresponding to the codeword $c\in\cC$.
\end{definition}

Below, we list several specific random row operations that fit into our framework.   
\begin{enumerate}
\item {\em Random Sampling:} Let $\Sigma = \Sigma_0$ be any alphabet, and let $\cD=\left(\cU_r\right)^n$, where $\cU_r$ is the uniform distribution on the $n_0$ basic row operations $r^{\text{(ip)}}_{e_j}$ for $j \in [n_0]$, where $e_j$ is the $j^{th}$ standard basis vector.  Thus, each row of $\cC$ is a row of $\cC_0$, chosen independently uniformly with replacement.
\item {\em Random Puncturing:} Same as above except $r_1,\dots,r_n$ are chosen {\em without} replacement.
\item {\em Random $t$-wise XOR:} Let $\Sigma_0=\Sigma=\F_2$ and $\cD=\left(\cU_{\oplus,t}\right)^n$. $\cU_{\oplus,t}$ is the uniform distribution over the $\binom{n_0}{t}$ basic row operations 
\[ \inset{ r^{\text{(ip)}}_{v} \suchthat v \in \F_2^{n_0} \text{ has weight } t }. \]
That is, to create a new row of $\cC$, we choose $t$ positions from $\cC_0$ and XOR them together.
\item {\em Random $t$-wise aggregation:} 
Let $\Sigma = \Sigma_0^t$, for any alphabet $\Sigma_0$, and let $\cD =\left(\dprod{t}\right)^n$, where $\dprod{t}$ is the uniform distribution over the $\binom{n_0}{t}$ basic row operations
\[ \inset{ r^{\text{(agg)}}_S : S\subset [n_0], |S| = t }. \]
\item {\em Random $t$-wise folding:} Let $\Sigma = \Sigma_0^t$, for any alphabet $\Sigma_0$.  For each partition $\pi = (S_1,\ldots,S_{n_0/t})$ of $[n_0]$ into sets of size $t$, consider the row operation $f_{\pi} = (r_1,\ldots,r_n)$ where
\[ r_j = r^{\text{(agg)}}_{S_j}. \]
Let $\cD$ be the uniform distribution over $f_{\pi}$ for all partitions $\pi$.
\end{enumerate}
The following column operations also fit into this framework; in this paper, we consider only the first.  We mention the second operation (random interleaving) in order to parallel the situation with columns.  We leave it as an open problem to study the effect of interleaving.
\begin{enumerate}
\item {\em Random sub-code:} Let $\Sigma = \Sigma_0$ be any alphabet, and let $\cD=\left(\cU_c\right)^N$, where $\cU_c$ is the uniform distribution on the $N_0$ basic column  operations 
\[ \inset{ c^{\text{(ip)}}_w \suchthat w = e_i, i\in[N_0] }. \]
That is, $\cC$ is formed from $\cC_0$ by choosing codewords independently, uniformly, with replacement from $\cC_0$.  

Notice that if $\cC_0$ is a linear code over $\F_q$, then this operation is the same if we replace $\inset{ w = e_i \suchthat i \in [N_0] }$ with all of $\F_q^n$, or with all vectors of a fixed weight, etc.  Thus, we do not separately consider random XOR (or inner products), as we do with columns.
\item {\em Random $t$-wise interleaving:} In this case $\cD=\left(\dprod{t}^c\right)^n$. $\dprod{t}^c$ is the uniform distribution over the $\binom{N_0}{t}$ basic column operations
\[ \inset{ c^{\text{(agg)}}_T : T\subset [N_0], |T| = t }. \]
\end{enumerate}


\section{Overview of Our Techniques}
\label{sec:techniques}

\paragraph{Random Row Operations.} In addition to answering Questions \ref{ques:enc} and \ref{ques:fold},
one of the contributions of this work is to exhibit the generality of
the techniques developed in~\cite{rw14}.  As such, our proofs follow their framework.
In that work, there were two steps: the first step was to bound the list-decodability in expectation (this will be defined more precisely below), and the second step was to bound the deviation from the expectation.  In this work, we use the deviation bounds as a black box, and it remains for us to bound the expectation.  We would also like to mention that we could have answered Questions \ref{ques:enc} and \ref{ques:fold} by applying the random puncturing results from~\cite{woot2013,rw14} as a black box to the XOR and direct product of the original code. We chose to unpack the proof to illustrate the generality of the proof technique developed in~\cite{woot2013,rw14} (and they also seem necessary to prove the generalization to the operation of taking random linear combinations of the rows of the code matrix).

The results on random row operations in this paper build on the approaches of~\cite{woot2013, rw14}.  While those works are aimed at specific questions (the list-decodability of random linear codes and of Reed-Solomon codes with random evaluation points), the approach applies more generally.  In this paper, we interpret the lessons of~\cite{woot2013, rw14} as follows: 
\begin{quote}
If you take a code over $\Sigma_0$ that is list-decodable (enough) up to $\rho_0 = 1 - 1/|\Sigma_0| - \eps$, and do some random (enough) stuff to the symbols, you will obtain a new code (possibly over a different alphabet $\Sigma$) which is list-decodable up to $\rho = 1 - 1/|\Sigma| - O(\eps)$.   If the random stuff that you have done happens to, say, increase the rate, then you have made progress.
\end{quote}

First, our notion of a random row operation $\cD$ being random enough is the same as $\cD$ having independent symbols (or independent symbols without replacement).
Now, we will quantify what it means to be ``list-decodable  enough" in the setup described above.  We introduce a parameter $\cE = \mathcal{E}(\cC_0,\cD)$, defined as follows:
\begin{equation}\label{eq:curlyE}
\mathcal{E}(\cC_0,\cD) := \max_{\Lambda \subset \cC_0, |\Lambda| = L} \EE_{f \sim \cD} \max_{z \in \Sigma^n} \sum_{c \in \cC_0} \agr(f(c), z).
\end{equation}
The quantity $\cE$ captures how list-decodable $\cC$ is in expectation.  Indeed, $\max_z \sum_{c \in \cC_0} \agr(f(c),z)$ is the quantity controlled by average-radius list-decodability (Definition \ref{def:avgrad}).  To make a statement about the actual average-radius list-decodability of $\cC$ (as opposed to in expectation), we will need to understand $\cE$ when the expectation and the maximum are reversed:
\[ \EE_{f \sim \cD} \max_{\Lambda \subset \cC_0, |\Lambda| = L} \max_{z \in \Sigma^n} \sum_{c \in \cC_0} \agr(f(c), z).\]
The work of \cite{woot2013,rw14} shows the following theorem.
\begin{theorem}\label{thm:mainthm}
Let $\cC_0, \cD$ and $\cC$ be as above, and suppose that $\cD$ has independent symbols.  Fix $\eps > 0$.
Then
\[ \EE_f \max_{z\in \Sigma^n} \max_{\Lambda \subset \cC_0, |\Lambda| = L} \sum_{c \in \Lambda} \agr(f(c),z) \leq 
\mathcal{E} + Y + \sqrt{ \mathcal{E} Y },
\]
where
\[ Y = C L \log(N) \log^5(L)\]
for an absolute constant $C$.
For $|\Sigma| = 2$, we have
\[ \EE_f \max_{x \in \Sigma^n} \max_{\Lambda \subset \cC_0, |\Lambda| = L} \sum_{c \in \Lambda} \agr(f(c),z)
\leq \mathcal{E} + C L \sqrt{n \ln(N) }. \]
\end{theorem}
Theorem \ref{thm:mainthm} makes the intuition above more precise:  Any ``random enough" operation (that is, an operation with independent symbols) of a code with good ``average-radius list-decodability" (that is, good $\cE(\cC_0,\cD)$) will result in a code which is also list-decodable.
In Appendix \ref{app:norepl}, we show that Theorem \ref{thm:mainthm} in fact implies the same result when ``random enough" is taken to be mean that $\cD$ has symbols drawn independently at random instead:
\begin{cor}\label{cor:norepl}
Theorem \ref{thm:mainthm} holds when ``independent symbols" is replaced by ``symbols drawn independently without replacement".
\end{cor}
In this work, we answer Questions \ref{ques:enc} and \ref{ques:fold} by coming up with useful distributions $\cD$ on functions $f$ and computing the parameter $\cE$.  To do this, we will make use of some average-radius Johnson bounds; we record these in Appendix~\ref{app:johnson-bound}.

\paragraph{Random Column Operations.} Our result on random subcodes follows from a simple probabilistic method. The argument for showing that the parameters in this positive result cannot be improved, we construct a specific code $\cC_0$. The code $\cC_0$ consists of various ``clusters", where each cluster is the set of all vectors that are close to some vector in another code $C^*$. The code $C^*$ has the property that it is list decodable from a large fraction of errors and that for smaller error rate its list size is suitably smaller-- the existence of such a code with exponentially many vectors follows from the standard random coding argument. This allows the original code $\cC_0$ to even have good average-radius list decodability. The fact that the cluster vectors are very close to some codeword in $C^*$ (as well as the fact that $C^*$ has large enough distance) basically then shows that the union bound used to prove the positive result is tight.


\section{General Results}
\label{sec:gen}
In this section, we state our results about the effects of some particular random operations---XOR, aggregation, and subcodes---on list-decodability.  In Section \ref{sec:app}, we will revisit these operations and resolve Questions~\ref{ques:enc},~\ref{ques:fold} and~\ref{ques:subcode}.
\subsection{Random $t$-wise XOR}
\label{sec:xor}
In this section, we consider the row-operation of $t$-wise XOR.  We prove the following theorem.
\begin{theorem}
\label{thm:gen-xor}
Let $\cC_0\in \F_2^{n_0\times N}$ be a code with distance $0<\delta_0<1/2$. Let $\cD = \inparen{ \cU_{\oplus,t}}^n$, as defined in Section~\ref{sec:setup}, and consider the code operation $f \sim \cD$.  Suppose that $t = 4\ln(1/\eps) \delta^{-1}_0$. Then for sufficiently small $\eps>0$ and large enough $n$, with probability $1-o(1)$, $\cC=f(\cC_0)$ is $(1/2(1-O(\eps)),\eps^{-2})$-average-radius list decodable and has rate $\Omega(\eps^2)$.
\end{theorem}

With the goal of using Theorem \ref{thm:mainthm}, 
we begin by computing the quantity $\cE(\cC_0, \cD)$.
\begin{lemma}\label{lem:xorE}
Let $\cC_0 \in \F_2^{n_0}$ be a code with distance $\delta_0$, and suppose $t \geq \frac{4\ln(1/\eps) }{\delta_0}$.  Then
\[ \cE(\cC_0, \mathcal{D} ) \leq \frac{n}{2} \inparen{ L(1 + \eps) + \sqrt{L} }.\]
\end{lemma}
The proof of Lemma~\ref{lem:xorE} follows from an application of an average-radius Johnson bound (see Appendix~\ref{app:johnson-bound} for more on these bounds).  The 
proof is given in Appendix~\ref{app:cE}.
Given Lemma~\ref{lem:xorE}, Theorem~\ref{thm:mainthm} implies that with constant probability, 
\begin{align*}
\max_{z\in \F_2^n} \max_{\Lambda \subset \cC, |\Lambda| = L} \frac{1}{L} \sum_{c \in \Lambda} \agr( c, z) &\leq 
\frac{\cE}{L} + C \sqrt{n\ln(N)} \\
&\leq \frac{n}{2} \inparen{1 + \eps + \frac{1}{\sqrt{L}}} + C \sqrt{n\ln N}. 
\end{align*}
In particular, if $C\sqrt{n \ln N} \leq \eps n$, 
then 
in the favorable case $\cC$ is $(\rho, L-1)$-average-radius list-decodable,
for $L = \eps^{-2}$ and $\rho = \frac{1}{2}\cdot( 1- C'\eps)$ for some constant $C'$.

It remains to verify the rate $R$ of $\cC$.  Notice that if $|\cC| = N$, then we are done, because then
the requirement $C\sqrt{n\ln(N)} \leq \eps n$ reads
\[ R = \frac{ \log_2(N) }{n} \leq \frac{ \eps^2}{C\ln(2)}. \]
Thus, to complete the proof we will argue that $f$ is injective with high probability, and so in the favorable case $|\cC| = N$.
Fix $c \neq c' \in \cC_0$. Then, by the same computations as in the proof of Lemma~\ref{lem:xorE},
\[\PR{f(c)=f(c')} =\left(\frac{1}{2}\left(1+(1-\delta_0)^t\right)\right)^n \le \left(\frac{1+\eps^2}{2}\right)^n. \]
Using the fact that we will choose $n \geq C \ln(N) /\eps^2$, the right hand side is
\[ \inparen{ \frac{ 1 + \eps^2 }{2} }^{C\ln(N)/\eps^2} = N^{- \ln \inparen{ \frac{2}{1 + \eps^2}} C / \eps^2 } \leq N^{-3} \]
for sufficiently small $\eps$. 
Thus, by the union bound on the $\binom{N}{2}\le N^2$ choices for the pairs of distinct codewords $(c,c')$, we see that $\PR{|\cC|<N} \le 1/N$, which is $o(1)$ as desired.
This completes the proof of Theorem \ref{thm:gen-xor}.

\begin{remark}[Random inner products for $q > 2$]
For our application (Question \ref{ques:enc}), $q=2$ is the interesting case.  However, the argument above goes through for $q> 2$.  In this case, we may use the first statement of Theorem \ref{thm:mainthm}, and statements 2 or 3 of Theorem \ref{thm:jb} for the average-radius Johnson bound.  
\end{remark}

\subsection{Random $t$-wise aggregation}
\label{sec:dp}

Theorem \ref{thm:fold} below analyzes $t$-wise aggregation in two parameter regimes.  In the first parameter regime,
we address Question \ref{ques:fold}, and we consider $t$-wise direct product where $n_0 = nt$.  
In this case, final code $\cC$ will have the same rate as the
original code $\cC_0$, and so in order for $\cC$ to be list-decodable up to
radius $1 - \eps$, the rate $R_0$ of $\cC_0$ must be $O(\eps)$. 
Item 1 shows that if this necessary condition is met (with some logarithmic slack), then $\cC$ is indeed list-decodable up to $1 - \eps$.
In the second parameter regime, we consider what can happen when the rate $R_0$ of $\cC_0$ is significantly larger.  In this case, we cannot hope to take $n$ as small as $n_0/t$ and hope for list-decodability up to $1 - \eps$.  The second part of Theorem \ref{thm:fold} shows that we may take $n$ nearly as small as the list-decoding capacity theorem allows.
\begin{theorem}\label{thm:fold} 
There are constants $C_i$, $i=0,\ldots,5$, so that the following holds.
Suppose $q > 1/\eps^2$.
Let $\cC_0 \subset \F_q^{n_0}$ be a code with distance $\delta_0 \geq C_2 > 0$. 

\begin{itemize}
\item[1.]
Suppose $t \geq C_0 \log(1/\eps)\ge 4\ln(1/\eps)/\delta_0$.
Suppose that $\cC_0$ has rate 
\[R_0 \leq \frac{C_1\eps}{\log(q) t \log^5(1/\eps)}.\]
Let $n = n_0/t$, and let 
$\cD_ = \inparen{ \dprod{t} }^n$ be the $t$-wise aggregation operation of Section~\ref{sec:setup}.  Draw $f\sim \cD$, and let
$\cC = f(\cC_0)$. 
Then with high probability, $\cC$ is $(1 - C_3\eps, 1/\eps)$-average-radius list-decodable, and further the rate $R$ of $\cC$ satisfies $R = R_0$.

\item[2.]
Suppose that $t \geq 4\ln(1/\eps)/\delta_0$, and
suppose that $\cC_0$ has rate $R_0$ so that
\[ R_0 \leq \inparen{ \frac{ nt }{n_0} } \inparen{ \frac{ \log(1/\eps) }{\log(q) } }.\]
Choose $n$ so that
\[ n \geq \frac{ \log(N) \log(1/\eps)}{\eps}. \]
Let
$\cD_ = \inparen{ \dprod{t} }^n$ be the $t$-wise aggregation operation of Section~\ref{sec:setup}.  Draw $f\sim \cD$, and let
$\cC = f(\cC_0)$. 
 Then with high probability, $\cC$ is $(1 - C_4 \eps, 1/\eps)$-average-radius list-decodable, and the rate $R$ of $\cC$ is at least
\[ R \geq \frac{ C_5\eps }{ t \log(q) \log^5(1/\eps) }. \]
\end{itemize}
\end{theorem}

The rest of this section is devoted to the proof of Theorem \ref{thm:fold}. 
As before, it suffices to control $\mathcal{E}(\cC_0, \cD)$.

\begin{lemma}\label{lem:foldcE}
With the set-up above, we have
\[ \cE(\cC_0, \cD) \leq Cn.\]
\end{lemma}
Again, the proof of Lemma \ref{lem:foldcE} follows from an average-radius Johnson bound.  The proof is given in Appendix \ref{app:cE}.
Then by Theorem \ref{thm:mainthm}, 
recalling that
\[ Y = C L \log(N) \log^5(L),\]
and $N = |\cC_0|$,
we have with high probability that
\begin{align*}
\EE_f \max_{z\in \Sigma^n} \max_{\Lambda \subset \cC_0, |\Lambda| = L} \sum_{c \in \Lambda} \agr(f(c),z) &\leq 
\mathcal{E}(\cC_0, \cD) + Y + \sqrt{ \mathcal{E}(\cC_0,\cD) Y }\\
&\leq O\inparen{  L\log(N) \log^5(L) + n }.
\end{align*}
In the favorable case,
\begin{equation}\label{eq:favorable}
\EE_f \max_{z \in \Sigma^n} \max_{\Lambda \subset \cC, |\Lambda| = L} \frac{1}{L} \sum_{c \in \Lambda} \agr(c, z) \leq O\inparen{ \log(N) \log^5(L) + n/L } = O\inparen{ \log(N)\log^5(1/\eps) + n\eps }.
\end{equation}
As before, $\cC$ is $(1 - C\eps, L-1)$ average-radius list-decodable, for some constant $C$, as long as the right hand side is no more than $O(n\eps)$.  This holds as long as 
\begin{equation}\label{eq:nreq}
\log(N) \log^5(1/\eps) \leq n\eps.
\end{equation}

Equation \eqref{eq:nreq} holds for any choice of $n$.  First, we prove item $1$ and we focus on the case that $n_0 = nt$; 
this mimics the parameter regime the definition of folding (which addresses Question~\ref{ques:fold}).
Given $n_0 = nt$, we can translate \eqref{eq:nreq} into a condition on $R_0$, the rate of $\cC_0$.
We have
\[R_0  = \frac{ \log_q(N) }{ n_0 } = \frac{ \log_q(N) }{ nt }, \]
and so translating \eqref{eq:nreq} into a requirement on $R(\cC_0)$, we see that as long as
\[ R_0 \lesssim \frac{ \eps }{ \log(q) t \log^5(1/\eps) } \lesssim \frac{ \eps }{\log(q) \log^6(1/\eps) }, \]
then with high probability $\cC$ is $(1 - C\eps, L)$-list-decodable.
Choose $n$ so that this holds. 
It remains to verify that the rate $R$ of $\cC$ is the same as the rate $R_0$ of $\cC_0$.  The (straightforward) proof is deferred to Appendix~\ref{app:claim:rate}.
\begin{claim}\label{claim:rate} 
With $\cC_0$ as above and with $n_0 = nt$, 
$|\cC| = N$ with probability at least $1 - o(1)$.
\end{claim}

By a union bound, with high probability both the favorable event \eqref{eq:favorable} occurs, and Claim \ref{claim:rate} holds.  In this case, $\cC$ is $(1 - C\eps, L)$-list-decodable, and the rate $R$ of $\cC$ is
\[ R = R_0. \]

Next, we consider Item 2, where we may choose $n < n_0/t$, thus increasing the rate.  It remains true that as long as \eqref{eq:nreq} holds, then $\cC$ is $(1 - C\eps, L)$-list-decodable.  Again translating the condition \eqref{eq:nreq} into a condition on $\log_{q^t}(N)/n$, we see that as long as
\begin{equation}\label{eq:choosen}
 \frac{ \log_{q^t}(N) }{n} \leq \frac{ \eps }{ t \log(q) \log^5(1/\eps) },
\end{equation}
then $\cC$ is $(1 - C\eps, L)$-list-decodable.
Now we must verify that the left-hand-side of \eqref{eq:choosen} is indeed the rate $R$ of $\cC$, that is, that $|\cC| = N$. As before, the proof is straightforward and is deferred to Appendix~\ref{app:claim:rate2}.
\begin{claim}\label{claim:rate2} With $\cC_0$ as above and with $n$ arbitrary, $|\cC| = N$ with probability at least $1 - o(1)$.
\end{claim}

Now, recalling our choice of $n$ in \eqref{eq:choosen}, with high probability both \eqref{eq:favorable} occurs and Claim \ref{claim:rate2} holds.  In the favorable case,
$\cC$ is $( 1- C\eps, L)$-list-decodable, 
as long as the rate $R$ satisfies
\[ R = \frac{ \log_{q^t}(|\cC|) }{n} = \frac{\log_{q^t}(N) }{n} \leq \frac{ C \eps }{ t\log^5(1/\eps) \log(q) }. \]
This completes the proof of Theorem \ref{thm:fold}.

\subsection{Random sub-codes}
In this section we address the case of random sub-codes.  Unlike the previous sections, the machinery of~\cite{rw14, woot2013} does not apply, and so we prove the results in this section directly.
We have the following proposition.
\begin{proposition}
\label{prop:random-subcode}
Let $\cC_0$ be any $(\rho,L_0)$-list decodable $q$-ary code. Let $\cC$ be a random sub-code of $\cC_0$ with $N=pN_0$ (as in the definition in Section \ref{sec:setup}), where
\[p=\frac{1}{q^{\eps n}\cdot L_0}.\]
With probability $1-o(1)$, the random subcode $\cC$ is $\left(\rho,\frac{3}{\eps}\right)$-list decodable. Further, the number of distinct columns n $\cC$ is at least $pN_0/2$.
\end{proposition}
The proof of Proposition~\ref{prop:random-subcode} follows straightforwardly from some Chernoff bounds.  We defer the proof to Appendix~\ref{app:random-subcode-pf}.

\begin{remark}
In Proposition~\ref{prop:random-subcode}, the choice of $3/\eps$ for the final list size was arbitrary in the sense that the $3$ can be made arbitrarily close to $1$ (assuming $\eps$ is small enough).
\end{remark}

Proposition~\ref{prop:random-subcode} only works for the usual notion of list decodability. It is natural to wonder if a similar result holds for average-radius list decodability. We show that such a result indeed holds (though with slightly weaker parameters) in Appendix~\ref{app:random-subcode}. 

It is also natural to wonder if one can pick a larger value of $p$---closer to $1/L_0$ than to $1/(q^{\eps n} L_0)$---in the statement of Proposition~\ref{prop:random-subcode}.
In particular, if $L_0$ is polynomial in $n$, could we pick $p=q^{-o(\eps n)}$?  In Appendix~\ref{app:random-subcode}, we show that this is not in general  possible.  More precisely, we show the following theorem.
\begin{theorem}
\label{thm:subcode-lb}
For every $\rho > 0$, and for every $0 < \alpha < \frac{1 - \rho}{12}$, and for every $n$ sufficiently large, there exists a  code $\cC_0$ with block length $n$ that is $(\rho,n)$-average-radius list decodable such that the following holds. Let $\cC$ be  obtained by picking a random sub-code of $\cC_0$ of size $N = pN_0$ where  $p=q^{-\alpha n}/n$. Then with high probability if $\cC$ is $(\rho',L)$-list decodable for any $\rho'\ge 1/n$, then $L\ge \Omega(1/\alpha)$.
\end{theorem}



\section{Applications}
Finally, we use the results of Section~\ref{sec:gen} to resolve Questions~\ref{ques:enc},~\ref{ques:fold}, and~\ref{ques:subcode}.
\label{sec:app}
\subsection{Linear time near optimal list decodable codes}
\label{sec:xor-app}
First, we answer Question \ref{ques:enc}, and give linear-time encodable binary codes with the optimal trade-off between rate and list-decoding radius.  
Our codes will work as follows.  We begin with a linear-time encodable code with constant rate and constant distance; we will use Spielman's variant on expander codes~\cite[Theorem 19]{spielman}.  These codes have rate $1/4$, and distance $\delta_0 \geq 0$ (a small positive constant).  Notice that a random puncturing of $\cC_0$ (as in~\cite{woot2013, rw14}) will not work, as $\cC_0$ does not have good enough distance---however, a random XOR, as in Section \ref{sec:xor} will do the trick.

\begin{cor}\label{thm:enc} There is a randomized construction of binary codes $\cC \in \F_2^n$ so that the following hold with probability $1 - o(1)$, for any sufficiently small $\eps$ and any sufficiently large $n$.
\begin{enumerate}
\item $\cC$ is encodable in time $O(n \ln(1/\eps))$.
\item $\cC$ is $(\rho, L)$-average-radius list-decodable with $\rho = \frac{1}{2}(1 - C \eps)$ and $L = \eps^{-2}$, where $C$ is an absolute constant.
\item $\cC$ has rate $\Omega(\eps^2)$. 
\end{enumerate}
\end{cor}

Indeed, let $\cC_0$ be as above.
Let $t = 4\ln(1/\eps)\delta_0^{-1}$, and choose $f \sim \inparen{ \cU_{\oplus, t} }^n$, as in Theorem \ref{thm:gen-xor}.
Let $\cC = f(\cC_0)$.  Items 2. and 3. follow immediately from 
Theorem~\ref{thm:gen-xor}, so it remains to 
verify Item 1 of Theorem \ref{thm:enc}, that $\cC$ is linear-time encodable.  Indeed, we have
\[ \cC(x) = A \cC_0(x), \]
where $A \in \F_2^{n \times n_0}$ is a matrix whose rows are binary vectors with at most $t$ nonzeros each. 
In particular, the time to multiply by $A$ is $nt = O(n \ln(1/\eps))$, as claimed.

\subsection{Random Folding}
\label{sec:fold}

Next, we further discuss Question \ref{ques:fold}, which asked for a rigorous version of the intuition behind results for folded Reed-Solomon codes and expander-based symbol aggregation.  
The intuition is that increasing the alphabet size effectively reduces the number of error patterns a decoder has to handle, thus making it easier to list-decode.  To make this intuition more clear, consider the following example when $q=2$.
Consider an error pattern that corrupts a $1-2\eps$ fraction of the \em odd \em positions (the rest do not have errors).  This error pattern must be handled by any decoder which can list decode from $1/2 - \eps$ fraction of errors.  On the other hand, consider a $2$-folding (with partition as above) of the code; now the alphabet size has increased, so we hope to correct $1-1/2^2-\eps=3/4-\eps$ fraction of errors.  However, the earlier error pattern affects a $1-2\eps$ of the new, folded symbols.  Thus, in the folded scenario, an optimal decoder need not handle this error pattern, since $1 - 2\eps > 3/4 - \eps$ (for small enough $\eps$). 

In Theorem \ref{thm:fold}, Item 1, we have shown that if $\cC_0$ is any code with distance bounded away from $0$ and with rate sufficiently small (slightly sublinear in $\eps$), has abundant random $t$-wise aggregation of symbols which are list-decodable up to a $1 -\eps$ fraction of errors, when $n = n_0/t$ and $t$ is large enough (depending only on $\eps$ and $q$).   This is the same parameter regime as folded Reed-Solomon codes (up to logarithmic factors in the rate), and thus the Theorem answers Question~\ref{ques:fold} insofar as it lends a rigorous way to interpret $t$-wise aggregation in this parameter regime.

\begin{remark}
While the intuition above applies equally well to folding and more general $t$-wise symbol aggregation,
We note that a random folding and a random symbol aggregation are not the same thing.  In the latter, the symbols of the new code may overlap, while in the former they may not.  However, allowing overlap makes our computations simple; since the goal was to better understand the intuition above, we have done our analysis for the simpler case of $t$-wise symbol aggregation.  It is an interesting open question to find a (clean) argument for the folding operation, perhaps along the lines of the argument of Corollary~\ref{cor:norepl} for puncturing vs. sampling.
\end{remark}

\subsection{Applications of random sub-codes}

Finally, we observe that Proposition~\ref{prop:random-subcode} immediately answers Question~\ref{ques:subcode} in the affirmative.
Indeed, suppose that $\cC_0$ is $(\rho_0, L_0)$-list-decodable with rate $R_0$.  Then Proposition~\ref{prop:random-subcode} implies that with high probability, for any sufficiently small $\eps$, a random subcode of rate
\[ R_0 - O \inparen{ \eps \log(q) + \frac{\log(L_0)}{n}} \]
is $(\rho_0, 3/\eps)$-list-decodable.  In particular, if we start out with a binary code with constant rate and large but subexponential list size, the resulting subcode will also have constant rate, and constant list size. 

For example, this has immediate applicatons for Reed-Solomon codes. 
Guruswami and Xing~\cite{GX13} showed that for every real $R$, $0<\eps<1-R$ and prime power $q$, there is an integer $m>1$ such that Reed-Solomon codes defined over $\F_{q^m}$ with the evaluation points being $\F_q$ of rate $R$ can be list decoded from the optimal $1-R-\eps$ fraction of errors with list size $N^{\eps}$. Thus, Proposition~\ref{prop:random-subcode} implies that random sub-codes of these codes are optimally list decodable (in all the parameters). We remark that this result also follows from the work of Guruswami and Xing~\cite{GX13}; our argument above is arguably simpler, but does not come with an algorithmic guarantee as results of~\cite{GX13} do.

Given Proposition~\ref{prop:random-subcode}, it is natural to ask about the list-decodability of the subcode $\cC$ when the error radius $\rho$ may be different than $\rho_0$. 
It turns out that this also follows from Proposition \ref{prop:random-subcode}:
below, we will use Proposition~\ref{prop:random-subcode} to argue that if a code $\cC_0$ is optimally list decodable for some fixed $\rho_0>0$ fraction of errors, then its random subcodes with high probability are optimally list decodable from $\rho$ fraction of errors for any $\rho_0\le \rho<1-1/q$.
Towards that end, we will make the following simple observation:
\begin{lemma}
\label{lem:eb}
Let $\cC$ be $(\rho,L)$-list decodable $q$-ary code. Then for every $\rho\le \rho'<1-1/q$, $\cC$ is also $(\rho',L')$-list decodable, where
\[ L'\le L\cdot q^{n(H_q(\rho')-H_q(\rho)+o(1))}\cdot 2^n.\]
\end{lemma}
\begin{proof}
Consider a  received word $y\in [q]^n$ such that $|\cC\cap B_q(y,\rho'n)|=L'$. Now we claim that there exists a $z\in B_q(y,\rho'n)$ such that 
\begin{align}
\label{eq:exists-y}
|B_q(z,\rho n)\cap \cC| &\ge L'\cdot \frac{(q-1)^{\rho n}}{|B_q(y,\rho' n)|}\\
\label{eq:exists-y-bound}
& \ge L'\cdot \frac{q^{H_q(\rho)n-o(n)}}{2^n}\cdot \frac{1}{q^{H_q(\rho' n)}}.
\end{align}
In the above the second inequality follows from the following facts:
volume of $q$-ary Hamming balls of radius $\gamma n$ are bounded from above by $q^{H_q(\gamma)n}$ and from below by $q^{H_q(\gamma) n -o(n)}$ (and that $\binom{n}{\rho n} (q-1)^{\rho n} \ge q^{H_q(\rho) n -o(n)}$). 
\eqref{eq:exists-y-bound}  along with the fact that $\cC$ is $(\rho, L)$-list decodable
proves the claimed bound on $L'$. 

To complete the proof we argue~\eqref{eq:exists-y}: we show the existence of $z$ by the probabilistic method:\footnote{This part of the proof is similar to the argument used to prove the Elias-Bassalygo bound~\cite{book}.}  pick $z\in B_q(y,\rho' n)$ uniformly at random. Fix a $c\in\cC\cap B_q(y,\rho' n)$. Then
\[\PR{ c\in B_q(z,\rho n)} = \frac{|B_q(c,\rho n) \cap B_q(y,\rho' n)|}{B_q(y,\rho' n)}.\]
Next we argue that
\begin{equation}
\label{eq:c-cap-y-vol}
|B_q(c,\rho n) \cap B_q(y,\rho' n)| \ge (q-1)^{\rho n}.
\end{equation}
Note that the above implies that
\[\E{|B_q(z,\rho n)\cap \cC|} \ge L'\cdot \frac{(q-1)^{\rho n}}{|B_q(y,\rho' n)|},\]
which would prove~\eqref{eq:exists-y}. To see why~\eqref{eq:c-cap-y-vol} is true, consider any $\rho n$ positions where $c$ and $y$ agree on. Note that if we change all of those values (to any of the $(q-1)^{\rho n}$ possibilities) to obtain $c'$, then we have $d(c',y)\le \rho' n$ and $d(c',c)= \rho n$, which proves~\eqref{eq:c-cap-y-vol}.
\end{proof}

Lemma~\ref{lem:eb} along with Proposition~\ref{prop:random-subcode} implies the following.
\begin{cor}\label{cor:optsub}
Let $q\ge 2^{1/\eps}$.
Let $\cC_0$ be a $(\rho,L)$-list decodable $q$-ary code with optimal rate $1-H_q(\rho)-\eps$. Then for any $\rho'\ge \rho$, with probability at least $1 - o(1)$, a random subcode $\cC$ of $\cC_0$ of rate $1 - H_q(\rho') - O(\eps)$ is  $(\rho',O(1/\eps))$-list decodable.
\end{cor}

\begin{remark}
\label{rem:eb}
The bound in Lemma~\ref{lem:eb} is tight up to the $q^{o(n)}\cdot 2^n$ factor. In particular, one cannot have a bound of $L\cdot q^{\gamma n}$ for any $\gamma < H_q(\rho')-H_q(\rho)$ since that would contradict the list decoding capacity bounds.
\end{remark}





\section{Open Questions}
In this work we have made some (modest) progress on understanding on how random row and column operations change the list decodability of codes. We believe that our work highlights many interesting open questions. We list some of our favorites below:
\begin{enumerate}
\item Theorem~\ref{thm:fold} is proved for random $t$-wise direct product codes. It would be nice to prove the analog of item $1$ in Theorem~\ref{thm:fold} for random $t$-wise folding so that we can formally answer Question~\ref{ques:fold} in the affirmative.
\item We did not present any results for random $t$-wise interleaving. Gopalan, Guruswami and Raghavendra have shown that for any code $\cC_0$ its $t$-wise interleaved code $\cC$ (that is the code that deterministically applies all possible basic column operations that bunch together the $\binom{N_0}{t}$ subsets of columns of size $t$) the list decodability does not change by much~\cite{GGR11}. In particular, they show that if $\cC_0$ is $(\rho,L)$-list decodable then $\cC$ is $(\rho,L^{O(1)})$-list decodable. However, for {\em random} $t$-wise interleaving the list decoding radius might actually improve.\footnote{If this were to be the case then this could formalize the reason why the Parvaresh-Vardy codes~\cite{PV05}, which are sub-codes of interleaving of Reed-Solomon codes, have good list decodability properties.} We leave open the question of resolving this possibility.
\item As mentioned above, our work, and the results of Guruswami and Xing~\cite{GX13}, shows that random sub-codes of Reed-Solomon codes over $\F_{q^m}$ (for large enough $m$) with evaluation points from the sub-field $\F_q$ have optimal list decodable properties.  We believe that we should be able to derive such a result even if we start from any Reed-Solomon codes or at the very least if one starts off with a randomly punctured Reed-Solomon codes. Note that even though the results of~\cite{rw14} give near optimal list decodability results of Reed-Solomon codes, their results are logarithmic factors off from the optimal rate bounds.  Proposition~\ref{prop:random-subcode} implies that it suffices to prove a non-trivial exponential bound on the list size for list decoding rate $R$ Reed-Solomon codes from $1-R-\eps$ fraction of errors---a special case of this is proved in \cite{GX13}, but the general question is open.  
\item All of our results so far only use either just random row operation or just random column operations. An open question is to find applications where random row and column operations could be use together to obtain better results than either on their own.  The above point would be such an example, if resolved. 
\end{enumerate}

\subsection*{Acknowledgments}
We thank Swastik Kopparty and Shubhangi Saraf for initial discussions on Questions~\ref{ques:enc} and~\ref{ques:fold}  (and for indeed suggesting the random XOR as an operation to consider) and \href{http://www.dagstuhl.de/en/program/calendar/semhp/?semnr=12421}{Dagstuhl} for providing the venue for these initial discussions.
We thank Venkat Guruswami for pointing out the argument in Appendix~\ref{app:binary-lin}. Finally, we thank Parikshit Gopalan for pointing the connection of our results to existing results on XOR and direct product codes.  MW also thanks the theory group at IBM Almaden for their hospitality during part of this work.
\bibliographystyle{alpha}
\bibliography{refs}

\begin{thebibliography}{IJKW10}

\bibitem[Bli86]{blinovski}
Volodia~M. Blinovsky.
\newblock Bounds for codes in the case of list decoding of finite volume.
\newblock {\em Problems of Information Transmission}, 22(1):7--19, 1986.

\bibitem[Bli05]{blinovsky-q-1}
V.~M. Blinovsky.
\newblock Code bounds for multiple packings over a nonbinary finite alphabet.
\newblock {\em Probl. Inf. Transm.}, 41(1):23--32, 2005.

\bibitem[Bli08]{blinovsky-q-2}
V.~M. Blinovsky.
\newblock On the convexity of one coding-theory function.
\newblock {\em Probl. Inf. Transm.}, 44(1):34--39, 2008.

\bibitem[CGV13]{cgv2012}
Mahdi Cheraghchi, Venkatesan Guruswami, and Ameya Velingker.
\newblock Restricted isometry of fourier matrices and list decodability of
  random linear codes.
\newblock In {\em Proceedings of the Twenty-Fourth Annual ACM-SIAM Symposium on
  Discrete Algorithms (SODA)}, pages 432--442, 2013.

\bibitem[DL12]{DL12}
Zeev Dvir and Shachar Lovett.
\newblock Subspace evasive sets.
\newblock In {\em Proceedings of the 44th Symposium on Theory of Computing
  Conference (STOC)}, pages 351--358, 2012.

\bibitem[Eli57]{elias}
Peter Elias.
\newblock List decoding for noisy channels.
\newblock {\em Technical Report 335, Research Laboratory of Electronics, MIT},
  1957.

\bibitem[GGR11]{GGR11}
Parikshit Gopalan, Venkatesan Guruswami, and Prasad Raghavendra.
\newblock List decoding tensor products and interleaved codes.
\newblock {\em SIAM J. Comput.}, 40(5):1432--1462, 2011.

\bibitem[GI01]{GI01}
Venkatesan Guruswami and Piotr Indyk.
\newblock Expander-based constructions of efficiently decodable codes.
\newblock In {\em Proceedings of the 42nd Annual IEEE Symposium on the
  Foundations of Computer Science (FOCS)}, pages 658--667. IEEE, 2001.

\bibitem[GI03]{GI03}
Venkatesan Guruswami and Piotr Indyk.
\newblock Linear time encodable and list decodable codes.
\newblock In {\em Proceedings of the 35th Annual ACM Symposium on Theory of
  Computing (STOC)}, pages 126--135, 2003.

\bibitem[GI05]{GI05}
Venkatesan Guruswami and Piotr Indyk.
\newblock Linear-time encodable/decodable codes with near-optimal rate.
\newblock {\em IEEE Transactions on Information Theory}, 51(10):3393--3400,
  2005.

\bibitem[GK13]{GK13}
Venkatesan Guruswami and Swastik Kopparty.
\newblock Explicit subspace designs.
\newblock In {\em FOCS}, 2013.
\newblock To appear.

\bibitem[GN13]{gurnar2013}
Venkatesan Guruswami and Srivatsan Narayanan.
\newblock Combinatorial limitations of average-radius list decoding.
\newblock {\em RANDOM}, 2013.

\bibitem[GR08]{GR08}
Venkatesan Guruswami and Atri Rudra.
\newblock Explicit codes achieving list decoding capacity: Error-correction
  with optimal redundancy.
\newblock {\em IEEE Transactions on Information Theory}, 54(1):135--150, 2008.

\bibitem[GR09]{GR09}
Venkatesan Guruswami and Atri Rudra.
\newblock Error correction up to the information-theoretic limit.
\newblock {\em Commun. ACM}, 52(3):87--95, 2009.

\bibitem[GR10]{GR10}
Venkatesan Guruswami and Atri Rudra.
\newblock The existence of concatenated codes list-decodable up to the hamming
  bound.
\newblock {\em IEEE Transactions on Information Theory}, 56(10):5195--5206,
  2010.

\bibitem[GRS14]{book}
Venkatesan Guruswami, Atri Rudra, and Madhu Sudan.
\newblock Essential coding theory, 2014.
\newblock Draft available at
  \href{http://www.cse.buffalo.edu/~atri/courses/coding-theory/book/index.html}{http://www.cse.buffalo.edu/~atri/courses/coding-theory/book/index.html}.

\bibitem[Gur04]{venkat-thesis}
Venkatesan Guruswami.
\newblock {\em List Decoding of Error-Correcting Codes (Winning Thesis of the
  2002 ACM Doctoral Dissertation Competition)}, volume 3282 of {\em Lecture
  Notes in Computer Science}.
\newblock Springer, 2004.

\bibitem[Gur11]{G11}
Venkatesan Guruswami.
\newblock Linear-algebraic list decoding of folded reed-solomon codes.
\newblock In {\em IEEE Conference on Computational Complexity}, pages 77--85,
  2011.

\bibitem[GV10]{GV10}
Venkatesan Guruswami and Salil Vadhan.
\newblock A lower bound on list size for list decoding.
\newblock {\em Information Theory, IEEE Transactions on}, 56(11):5681--5688,
  2010.

\bibitem[GW13]{GW13}
Venkatesan Guruswami and Carol Wang.
\newblock Linear-algebraic list decoding for variants of reed-solomon codes.
\newblock {\em IEEE Transactions on Information Theory}, 59(6):3257--3268,
  2013.

\bibitem[GX12]{GX12}
Venkatesan Guruswami and Chaoping Xing.
\newblock Folded codes from function field towers and improved optimal rate
  list decoding.
\newblock In {\em Proceedings of the 44th Symposium on Theory of Computing
  Conference (STOC)}, pages 339--350, 2012.

\bibitem[GX13]{GX13}
Venkatesan Guruswami and Chaoping Xing.
\newblock List decoding reed-solomon, algebraic-geometric, and gabidulin
  subcodes up to the singleton bound.
\newblock In {\em Proceedings of the 45th ACM Symposium on the Theory of
  Computing (STOC)}, pages 843--852, 2013.

\bibitem[GX14]{GX14}
Venkatesan Guruswami and Chaoping Xing.
\newblock Optimal rate list decoding of folded algebraic-geometric codes over
  constant-sized alphabets.
\newblock In {\em Proceedings of the Twenty-Fifth Annual ACM-SIAM Symposium on
  Discrete Algorithms (SODA)}, pages 1858--1866, 2014.

\bibitem[IJKW10]{dp-codes}
Russell Impagliazzo, Ragesh Jaiswal, Valentine Kabanets, and Avi Wigderson.
\newblock Uniform direct product theorems: Simplified, optimized, and
  derandomized.
\newblock {\em SIAM J. Comput.}, 39(4):1637--1665, 2010.

\bibitem[Kop12]{K12}
Swastik Kopparty.
\newblock List-decoding multiplicity codes.
\newblock {\em Electronic Colloquium on Computational Complexity (ECCC)},
  19:44, 2012.

\bibitem[PV05]{PV05}
Farzad Parvaresh and Alexander Vardy.
\newblock Correcting errors beyond the guruswami-sudan radius in polynomial
  time.
\newblock In {\em Proceedings of the 46th Annual IEEE Symposium on Foundations
  of Computer Science (FOCS)}, pages 285--294, 2005.

\bibitem[Rud07]{atri-thesis}
Atri Rudra.
\newblock {\em List decoding and property testing of error-correcting codes}.
\newblock PhD thesis, University of Washington, 2007.

\bibitem[Rud11]{R11}
Atri Rudra.
\newblock Limits to list decoding of random codes.
\newblock {\em IEEE Transactions on Information Theory}, 57(3):1398--1408,
  2011.

\bibitem[RW14]{rw14}
Atri Rudra and Mary Wootters.
\newblock Every list-decodable code for high noise has abundant near-optimal
  rate puncturings.
\newblock In {\em Proceedings of the 46th annual ACM Symposium on the Theory of
  Computing (STOC)}, 2014.
\newblock To appear.

\bibitem[Spi96]{spielman}
Daniel~A. Spielman.
\newblock Linear-time encodable and decodable error-correcting codes.
\newblock {\em IEEE Transactions on Information Theory}, 42(6):1723--1731,
  1996.

\bibitem[Sud00]{madhu-survey}
Madhu Sudan.
\newblock List decoding: algorithms and applications.
\newblock {\em SIGACT News}, 31(1):16--27, 2000.

\bibitem[Tre03]{luca}
Luca Trevisan.
\newblock List-decoding using the xor lemma.
\newblock In {\em Proceedings of the 44th Symposium on Foundations of Computer
  Science (FOCS)}, pages 126--135, 2003.

\bibitem[Woo13]{woot2013}
Mary Wootters.
\newblock On the list decodability of random linear codes with large error
  rates.
\newblock In {\em Proceedings of the 45th annual ACM Symposium on the Theory of
  Computing (STOC)}, pages 853--860. ACM, 2013.

\bibitem[Woz58]{wozencraft}
John~M. Wozencraft.
\newblock List {D}ecoding.
\newblock {\em Quarterly Progress Report, Research Laboratory of Electronics,
  MIT}, 48:90--95, 1958.

\bibitem[ZP82]{ZP}
Victor~V. Zyablov and Mark~S. Pinsker.
\newblock List cascade decoding.
\newblock {\em Problems of Information Transmission}, 17(4):29--34, 1981 (in
  Russian); pp. 236-240 (in English), 1982.

\end{thebibliography}

\appendix

\section{Average case, average radius Johnson bounds}
\label{app:johnson-bound}

The \em Johnson bound \em states that any code with good enough distance is list-decodable with polynomial list sizes, up to a radius that depends on the distance.  For this work, we will need some slight variants on the Johnson bound.  We will be interested in average-radius list decoding, rather than the standard definition.  We state three versions of an average-radius Johnson bound below, for different list sizes.
\begin{theorem}[Average-radius Johnson bounds]\label{thm:jb}
Let $\mathcal{C}:\F_q^k \to \F_q^n$ be any code.  Then for all $\Lambda \subset \F_q^k$ of size $L$ and for all $z \in \F_q^n$:

\begin{itemize}
\item If $q=2$,
\[ \sum_{x \in \Lambda} \agr(\cC(x),z) \leq \frac{n}{2} \inparen{ L + \sqrt{ L^2 - 2\sum_{x \neq y \in \Lambda} d(\cC(x),\cC(y)) } }.\]
\item For all $\eps \in (0,1)$,
\[ \sum_{x \in \Lambda} \agr( \cC(x), z) \leq \frac{nL}{q} + \frac{ nL}{2\eps} \inparen{ 1 + \eps^2 }\inparen{1 - \frac{1}{q}} - \frac{n}{2L\eps}\sum_{x \neq y \in \Lambda} d(\cC(x),\cC(y)). \]
\item 
\[\sum_{x\in\Lambda} \agr(\cC(x),z) \le \frac{1}{2}\inparen{n+\sqrt{n^2+ 4n^2L(L-1)-4n^2\sum_{x\neq y\in\Lambda} d(\cC(x),\cC(y))}}.\]
\end{itemize}
\end{theorem}
\begin{proof}
The proof of the second two statements (for general $q$) can be found in~\cite{rw14}.  The statement for $q=2$ follows by the computation below (implicit in~\cite{woot2013,cgv2012}).  Let $\Phi \in (\pm 1)^{n \times 2^k}$ be the matrix whose columns are indexed by $x \in \F_2^k$, so that $\Phi_{j,x} = (-1)^{\cC(x)_j}$.  Let $\varphi_j$ denote the $j$-th column of $\Phi$.  Then
\begin{align*}
\max_z \sum_{x \in \Lambda} \agr(\cC(x), z) &= \sum_{j=1}^n \max_{b \in \{0,1\}} \sum_{x \in \Lambda} \ind{\cC(x)_j = \alpha} \\
&= \sum_{j=1}^n \max_{\alpha \in \{0,1\}} \sum_{x \in \Lambda} \frac{ (-1)^{\alpha}(-1)^{\cC(x)_j}  + 1 }{2} \\
&= \sum_{j=1}^n \inparen{ nL + \sum_{j=1}^n \inabs{ \ip{\varphi_j}{\ind{\Lambda}} } } \\
&= \frac{1}{2} \inparen{ nL + \norm{ \Phi \ind{\Lambda} }_1 } \\
&\leq \frac{1}{2} \inparen{ nL + \sqrt{n} \twonorm{ \Phi \ind{\Lambda}} },
\end{align*}
using Cauchy-Schwarz in the final line.  The claim then follows from the definition of $\Phi$
and the fact that the $(x,y)$-entry of $\Phi^T \Phi$ is given by $n(1 - 2d(\cC(x),\cC(y)))$
.   Indeed, from this, we have
\[ \twonorm{ \Phi \ind{\Lambda}}^2 =
\ind{\Lambda}^T \Phi^T \Phi \ind{\Lambda} =
n \sum_{x \in \Lambda} \sum_{y \in \Lambda} \inparen{1 - 2d(\cC(x),\cC(y))  },\]
and plugging this in above gives the statement.
\end{proof}

\section{Linear time encodable and decodable binary list decodable codes}
\label{app:binary-lin}

We will argue the following in this section:

\begin{theorem}
\label{thm:bin-lin}
For every $\eps>0$, there exists a binary code that can be encoded and list decoded in linear time from $1/2-\eps$ fraction of errors with rate $2^{-2^{O(\eps^{-9})}}$.
\end{theorem}

In the rest of the section, we argue why the statement above is true. (We thank Venkat Guruswami for pointing out the following argument to us.)

We will crucially use the following result that follows from the work of Guruswami and Indyk:
\begin{theorem}[\cite{GI03}]
\label{thm:gi-ld}
For every $\gamma>0$, there exists a $q$-ary code that can be encoded and decoded in linear time from $1-\gamma$ fraction of errors for $q=1/\gamma$ and with rate $2^{-2^{O(-\gamma^3)}}$.
\end{theorem}

Ultimately we will use the above theorem with $\gamma=\eps^3/8$ to get out outer code. Our inner code will be the binary Hadamard code with $q=8/\eps^3$ codewords in it. Since the binary Hadamard code has relative distance $1/2$, Johnson bound implies that it is $(1/2-\eps/2,8/\eps^2)$-list decodable. Our final code will be the code concatenation of the outer and inner code.

Note that the rate of the concatenated code is at least $1/q\cdot 2^{-2^{O(\eps^{-9})}}$, which is within the claimed bound on the rate. The claim on the encoding runtime follows from the fact that the outer code can be encoded in linear time and the inner code has constant size.

Finally, we look at the list decoding algorithm. The algorithm is simple:
\begin{enumerate}
\item Let $y=(y_1,\dots,y_N)$ be the received word where each $y_i$ is a valid received word for the inner code.
\item For each $i\in [N]$, compute the list of every message whose corresponding Hadamard codeword is within a relative Hamming distance of $1/2-\eps/2$ from $y_i$. Set $y'_i$ be a random element from this list of messages.
\item Run the list decoding algorithm for the outer code on the intermediate received word $(y'_1,\dots,y'_N)$.
\end{enumerate}

It is easy to check that the above algorithm runs in linear time since the list decoder for the outer code runs in linear time and inner code has constant size.

Finally, we argue why the above algorithm works. Consider any codeword that is within $1/2-\eps$ fraction of the received word. Then by an averaging argument, one can show that for at least $\eps$ fraction of the positions $i\in [N]$, the corresponding value in the outer codeword belong to the list calculated in Step 2 above. Since the list has size $8/\eps^2$, then in expectation the codeword agrees with the intermediate received word from Step 3 in $\eps^3/8$ fraction of positions. This implies that the list decoder from Theorem~\ref{thm:gi-ld} can recover the algorithm.\footnote{To be fully correct, we need to adjust the constants so that in expectation one has agreement in $\eps^3/4$ fraction of location since then with high probability one would indeed have agreement of at least $\eps^3/8$ for {\em all} codewords that need to be output. The latter is fine since it is known that the code from Theorem~\ref{thm:gi-ld} is actually $(1-\gamma,O(\gamma^{-3}))$-list decodable-- so a union bound would suffice.}

\section{With replacement vs. without replacement}
\label{app:norepl}

In this appendix, we show how to apply Theorem \ref{thm:mainthm} to operations like puncturing and folding, where the symbols do not quite have full independence.  
Our first lemma justifies the extension of Theorem \ref{thm:mainthm} to symbols which are sampled without replacement.
\begin{lemma}\label{lem:repl} Suppose that $f \sim \cD$ has symbols drawn independently without replacement from $\mathcal{S}_{\cD}$, as in Definition \ref{def:random-ops}.  
Let $\cD'$ be the corresponding distribution with replacement: that is, each $f_j$ is drawn i.i.d. uniformly at random from $\mathcal{S}_{\cD}$.
Then
\[ \EE_{f \sim \cD} \max_{z \in \Sigma^n} \max_{ \Lambda \subset \cC_0, |\Lambda| = L} \sum_{c \in \Lambda} \agr( f(c), z )
\leq
 \EE_{f \sim \cD'} \max_{z \in \Sigma^n} \max_{ \Lambda \subset \cC_0, |\Lambda| = L} \sum_{c \in \Lambda} \agr( f(c), z )\]
\end{lemma}
For example, suppose $f = (f_1,\ldots, f_n) \sim \cD$ is random puncturing, so $f_j(c) = c_{i_j}$ for a random subset $\inset{ i_1, \ldots, i_n } \subset [N]$ chosen uniformly without replacement.  Then $\cD'$ would be the \em random sampling \em operation of~\cite{rw14}.  That is, $f_j(c) = c_{i_j}$ chosen i.i.d. from $[N]$.  Thus, Lemma \ref{lem:repl} implies that the results of~\cite{rw14} for random sampling imply to random puncturing as well.

To prove Lemma \ref{lem:repl}, 
we will need to unpack the results of~\cite{rw14} a bit.  We introduce the following definition.
\begin{definition}\label{def:plurality}
For a set $\Lambda \subset \cC_0$, and an index $j \in [n]$, we define the \em plurality \em of the $j$'th symbol of $\cC_0$ in $\Lambda$ to be
\[ \pl_j(\Lambda) = \max_{\alpha \in \Sigma} \inabs{ \inset{ c \in \Lambda \suchthat f(c)_j = \alpha } }. \]
\end{definition}
Thus, $\pl_j(\Lambda)$ is a random variable, over the choice of $f \sim \mathcal{D}$. 
Further, we have
\[ \max_{z \in \Sigma^n} \sum_{c \in \Lambda} \agr( f(c), z) = \max_{c \in \Lambda} \sum_{j=1}^n \pl_j(\Lambda). \]
Thus, when $f \sim \cD'$ has independent symbols,
the random variables $\pl_j(\Lambda)$ are independent for different $j$.  When $f \sim \cD$ is independent with replacement, then we have a sum of independent random variables with replacement.  Thus, the following simple lemma will imply Lemma \ref{lem:repl}.
\begin{lemma}\label{lem:replaug}
Suppose that $X_1, \ldots, X_n$ are drawn without replacement from a finite set $\mathcal{S} \subset \R^d$ of size $N$.  Suppose that $Y_1, \ldots, Y_n$ are drawn independently and uniformly at random from $\mathcal{S}$.  Then
\[ \EE_X \infnorm{ \sum_{i=1}^n X_i } \leq \EE_Y \infnorm{ \sum_{i=1}^n Y_i }. \]
\end{lemma}
\begin{proof}
Consider the following distribution.  Draw $z_1,\ldots, z_N$ from a multinomial distribution with $n$ trials and event probabilities $p_i = 1/N$ for $i=1,\ldots, N$.  Let $\hat{z}'_i$ denote the $z_i$ sorted in decreasing order: notice that $\hat{z}_i = 0$ for all $i>n$.  
Draw a random permutation $\pi \sim S_n$ and definte $\hat{z}_i = \hat{z}'_{\pi(i)}$. Now we have $\sum_i \hat{z}_i = n$, and by symmetry, $\EE \hat{z}_i = 1$.
Now draw $X_1,\ldots,X_n$ and $Y_1,\ldots, Y_n$ from $\mathcal{S}$, as in the lemma statement.  Observe that the distribution of
\[ \sum_{i=1}^n \hat{z}_i X_i \]
is the same as the distribution of
\[ \sum_{i=1}^n Y_i. \]
In particular, we have
\begin{equation}\label{eq:expsame}
 \EE_{X,\hat{z}} \infnorm{ \sum_{i=1}^n \hat{z}_i X_i } = \EE_Y \infnorm{ \sum_{i=1}^n Y_i }. 
\end{equation}
On the other hand, we have 
\begin{equation}\label{eq:expineq}
 \EE_{X,\hat{z}} \infnorm{ \sum_{i=1}^n \hat{z}_i X_i } \geq \EE_X \infnorm{ \EE_{\hat{z}} \sum_{i=1}^n \hat{z}_i X_i } 
= \EE_X \infnorm{ \sum_{i=1}^n X_i }, 
\end{equation}
using the fact that $\EE_{\hat{z}} \hat{z}_i = 1$ for all $i = 1,\ldots, n$.
Together, \eqref{eq:expsame} and \eqref{eq:expineq} imply that
\[ \EE_X \infnorm{ \sum_{i=1}^n X_i } \leq \EE_Y \infnorm{ \sum_{i=1}^n Y_i }, \]
as desired.
\end{proof}
Now Lemma \ref{lem:replaug} implies Lemma \ref{lem:repl}.  Indeed, in Lemma \ref{lem:replaug}, we may take the vectors $Y_i \in \R^d$ for $d = {N \choose L}$ to be given by
\[ (Y_j)_\Lambda = \pl_j(\Lambda). \]

\section{Missing Proofs from Section~\ref{sec:gen}}
\subsection{Controlling the parameter $\cE$}\label{app:cE}
In this section, we show how to control the parameter $\cE$ for random $t$-wise XOR and for random $t$-wise aggregation, using the average-radius Johnson bound, Theorem~\ref{thm:jb}.
\begin{proof}[Proof of Lemma \ref{lem:xorE}]
We will use the average-radius Johnson bound, Theorem~\ref{thm:jb}.  Thus, we start by computing the expected distance between two symbols of the code $\cC \in \F_2^n$ obtained from $\cC_0$ and $\cD$. 
Let $c, c'$ denote two distinct codewords in $\cC_0$.
Recall that $\cU_{\oplus, t}$ is the uniform distribution over
\[ \inset{ r_{v}^{\text{(ip)}} \suchthat v\in \F_2^N \text{ has weight } t }, \]
and write $f = (r_1, \ldots, r_n)$.  Let $v_i \in \F_2^N$ denote the vector picked by the row operation $r_i$; thus, $v_i \in \F_2^N$ are chosen i.i.d. uniformly at random (with replacement).
  Then
\begin{align*}
\EE \delta( f(c), f(c') )  &= \frac{1}{n} \sum_{i=1}^n \PR{ f_i(c) \neq f_i(c') }\\
&= \PR{ \ip{v_i}{c} \neq\ip{v_i}{c'} }\\
&= \frac{1}{2}  \PR{ (c-c')_{\supp{v_i}} \neq 0 } \\
&= \frac{1}{2} \inparen{ 1 - (1 - \delta_0)^t }\\
&\leq \frac{1}{2} \inparen{ 1 - e^{-\delta_0 t/2} }.
\end{align*}
In particular, if $t = \frac{ 4\ln(1/\eps) }{\delta_0}$, then this is $\frac{1}{2} (1 - \eps^2)$.  Then Theorem \ref{thm:jb} implies that
{\allowdisplaybreaks
\begin{align*}
\cE(\cC_0, \cD_{ip}(t)) &= \max_{\Lambda \subset \cC_0} \EE_{f \sim \cD_{ip}(t)} \max_{z \in \F_2^n} \sum_{c \in \Lambda} \agr(f(c),z) \\
&\leq \max_{\Lambda} \EE_f \max_{z \in \F_2^n} \frac{n}{2} \inparen{ L + \sqrt{ L^2 - 2\sum_{c \neq c' \in \Lambda} \delta(f(c), f(c')) } }\\
&\leq \max_{\Lambda} \frac{n}{2} \inparen{ L + \sqrt{ L^2 - 2\sum_{c \neq c' \in \Lambda} \EE_f \delta(f(c), f(c')) } }\\
&\leq \frac{n}{2} \inparen{ L + \sqrt{ L^2 - 2\sum_{c \neq c' \in \Lambda} \frac{1}{2}( 1 - \eps^2) }}\\
&= \frac{n}{2} \inparen{ L + \sqrt{ L^2 \eps^2 + L(1 - \eps^2) } }\\
&\leq \frac{n}{2} \inparen{ L(1 + \eps) + \sqrt{L} }.
\end{align*}
}
\end{proof}
\begin{proof}[Proof of Lemma \ref{lem:foldcE}]
We wish to control $\cE(\cC_0,\cD)$,  which we do via the average-radius Johnson bound (Theorem \ref{thm:jb}).  Because we are interested in the parameter regime where $q \geq 1/\eps^2$, we use the third statement in Theorem \ref{thm:jb}.
Suppose $t \geq 4\ln(1/\eps) / \delta_0$ and set $L = 1/\eps$.
For $c \neq c' \in \cC_0$, we compute
\begin{align*}
 \EE_{f \sim \cD} \delta(f(c), f(c')) &= \frac{1}{n} \sum_{i=1}^n \PR{ f_j(c) \neq f_j(c') } \\
&= \PR{ \exists j \in S_i \suchthat c_j \neq c'_j } \\
&= 1 - (1 - \delta_0)^t \\
&\leq 1 - \eps^2, 
\end{align*}
using the choice of $t$ in the final line.
Thus, by Theorem \ref{thm:jb}, Item 3, 
\begin{align*}
\cE(\cC_0, \cD) &= \max_{\Lambda \subset \cC_0} \EE_{f \sim \cD_{dp}(t)} \max_{z \in \F_q^n} \sum_{c \in \Lambda} \agr(f(c),z) \\
&\leq  \max_{\Lambda \subset \cC_0} \EE_{f \sim \cD_{dp}(t)} \max_{z \in \F_q^n} \frac{1}{2} \inparen{ n + \sqrt{ n^2 + 4n^2L(L-1) - 4n^2 \sum_{c \neq c' \in \Lambda} \delta( f(c), f(c') ) }}\\
&=  \max_{\Lambda \subset \cC_0} \frac{1}{2} \inparen{ n + \sqrt{ n^2 + 4n^2L(L-1) - 4n^2 \sum_{c \neq c' \in \Lambda} \EE_f\delta( f(c), f(c') )}} \\
&\leq   \frac{1}{2} \inparen{ n + \sqrt{ n^2 + 4n^2L(L-1) - 4n^2 \sum_{c \neq c' \in \Lambda} (1 - \eps^2) }}\\
&=  \frac{n}{2} \inparen{ 1 + \sqrt{ 1 + 4L(L-1)\eps^2 }}\\
&\leq Cn,
\end{align*} 
using the choice of $L$ and defining $C = (1 + \sqrt{5})/2$.
\end{proof}

\subsection{Proof of Claim~\ref{claim:rate}}
\label{app:claim:rate}
\begin{proof}
The only way that $|\cC| < N$ is if two codewords $c \neq c' \in \cC_0$ collide, that is, if $f(c) = f(c')$.  This is unlikely: we have
\[ \PR{ f(c) = f(c') } = (1 - \delta_0)^{nt} \leq \eps^{2nt}. \]
By a union bound over ${N \choose 2} \leq N^2$ pairs $c \neq c'$, we conclude that the probability that $|\cC| < N$ is at most
\begin{equation}\label{eq:collide}
\PR{ |\cC| < N } \leq N^2 \eps^{2nt}.
\end{equation}
If $nt = n_0$, we have
\[ \PR{ |\cC| < N} \leq q^{2n_0 R_0} \eps^{2nt} = \inparen{ q^{R_0} \eps }^{2n_0}. \]
In particular, when $q^{R_0} < 1/\eps$, this is $o(1)$.  By our assumption, $R_0 < \eps$, and so this is always true for sufficiently small $\eps$.
\end{proof}

\subsection{Proof Claim~\ref{claim:rate2}}
\label{app:claim:rate2}

\begin{proof}
As in \eqref{eq:collide}, we have
\[ \PR{|\cC| < N } \leq N^2 \eps^{2nt}. \]
We may bound the right-hand-side by
\[ N^2 \eps^{2nt} = \inparen{ q^{R_0 n_0/n} \eps^{t} }^{2n}, \]
and for this to be $o(1)$, it is sufficient for
\[ R_0 \leq \inparen{ \frac{nt}{n_0} } \inparen{ \frac{ \log(1/\eps) }{\log(q)} },  \]
which was our assumption for part 2 of the theorem.
\end{proof}

\section{Missing details on random sub-codes}
\label{app:random-subcode}
\subsection{Preliminaries}
We collect some known results that we will use. We begin with a form of Chernoff bound that will be useful for our purposes:
\begin{theorem}
\label{thm:chernoff}
Let $X_1,\dots,X_m$ are random independent binary random variables with bias $p$. Then
\[\PR{\sum_i X_i > t} \le \left(\frac{pm}{t}\right)^{t-pm}.\]
\end{theorem}

Next, we state a conjecture concerning the tradeoff between list decodability and list size:
\begin{conjecture}
\label{conj:listsize}
Any $(\rho,L)$-list decodable $q$-ary code has rate at most $1-H_q(\rho)-\Omega\left(\frac{1}{L}\right)$.
\end{conjecture}
There many reasons to believe that the conjecture above is true. Conjecture~\ref{conj:listsize} is known to be true when $\rho$ approaches $1$~\cite{GV10,blinovsky-q-1,blinovsky-q-2,blinovski}. Weaker versions of the conjecture are known to be true.
\begin{theorem}[\cite{gurnar2013,blinovsky-q-1,blinovsky-q-2,blinovski}]
\label{thm:ld-lb}
For constant $\rho$,
any $q$-ary code that is $(\rho,L)$- list decodable must have rate at most $1-H_q(\rho)-\Omega\left(\frac{1}{2^L}\right)$.
\end{theorem}

\begin{theorem}[\cite{gurnar2013}]
\label{thm:avg-ld-lb}
For constant $\rho$,
any binary code that is $(\rho,L)$-average-radius list decodable must have rate at most $1-H_2(\rho)-\Omega\left(\frac{1}{L^2}\right)$.
\end{theorem}

Finally, the rate bound in Conjecture~\ref{conj:listsize} is achieved by random codes and the bound in Conjecture~\ref{conj:listsize} is known to be true for most codes~\cite{R11}.

\subsection{Proof of Proposition~\ref{prop:random-subcode}}\label{app:random-subcode-pf}
We now give the proof of Proposition~\ref{prop:random-subcode}.
\begin{proof}[Proof of Proposition~\ref{prop:random-subcode}]
Let $\Sigma$ be the alphabet of size $q$. Consider any fixed $y\in\Sigma^n$, where $n$ is the block length of $\cC_0$ (and is assumed to be large enough). Then the list decodability of $\cC_0$ implies that
\begin{equation}
\label{eq:L0}
\inabs{B_q(y,\rho n)\cap \cC_0} \le L_0,
\end{equation}
where $B_q(y,r)$ is the $q$-ary Hamming ball of radius $r$ centered at $y$. 
As in Section \ref{sec:setup}, write $\cC = f(\cC_0)$, where $f = (c_1,\ldots, c_N) \sim \inparen{ \mathcal{U}_c}^N$.
Now consider the random variable $\inabs{B_q(y,\rho n)\cap\cC}$, where we are abusing our chosen notation slightly and treating $\cC$ as a proper set, even though as a matrix, $\cC$ may have repeated columns.  
This is bounded by the sum of $N$ independent Bernoulli-$\left(\frac{\inabs{B_q(y,\rho n)\cap \cC_0}}{N_0}\right)$ variables:
\[ \inabs{ B_q(y, \rho n ) \cap \cC } \leq \sum_{i=1}^n \ind{ c_i( \cC_0 ) \in B_q(y, \rho n) }, \]
where again we have inequality rather than equality because of the possibility that $c_i(\cC_0) = c_j(\cC_0)$ for some $i \neq j$.
We have
\begin{equation}
\label{eq:avg-L}
\E{ \sum_{i=1}^n \ind{ c_i(\cC_0) \in B_q(y, \rho n) }}
= N\cdot \frac{\inabs{B_q(y,\rho n)\cap \cC_0}}{N_0} \le q^{-\eps n},
\end{equation}
where the last inequality follows from \eqref{eq:L0}.
Thus, by a Chernoff bound (Theorem~\ref{thm:chernoff}) along with \eqref{eq:avg-L},
\[\PR{\inabs{B_q(y,\rho n)\cap\cC} >\frac{3}{\eps}} \le 
\PR{  \sum_{i=1}^n \ind{ c_i(\cC_0) \in B_q(y, \rho n) } > \frac{3}{\eps}} \le
\left(\frac{\eps}{3\cdot q^{\eps n}}\right)^{3/\eps-q^{-\eps n}} \le \left(\frac{1}{q^{\eps n}}\right)^{\frac{2}{\eps}} =q^{-2n},\]
where the last inequality follows for large enough $n$. Taking the union bound over the $q^n$ choices of $y$, we conclude that $\cC$ is not $\left(\rho,\frac{3}{\eps}\right)$-list decodable with probability at most $q^{-n}$, which completes the proof.

The claim on the size of $\cC$ follows from the following simple argument. Note that $|\cC|< pN_0/2$ implies that there exists a subset $S\subset [N_0]$ of size exactly $pN_0/2$ such that all codewords in $\cC$ are contained in the columns of $\cC_0$ indexed by $S$. Note that the probability of this happening for a fixed $S$ is given by $(p/2)^{pN_0}$. Taking union bound over all choices of $S$, implies that the probability that $|\cC|<pN_0/2$ is upper bounded by
\[\binom{N_0}{pN_0/2}\cdot \left(\frac{p}{2}\right)^{pN_0}\le \left(\frac{2e}{p}\right)^{pN_0/2}\cdot \left(\frac{p}{2}\right)^{pN_0}= \left(\frac{ep}{2}\right)^{pN_0/2},\]
which is $o(1)$ by our choice of parameters. This completes the proof.
\end{proof}

\subsection{Upper Bound}

Proposition~\ref{prop:random-subcode} only works for the usual notion of list decodability. It is natural to wonder if a similar result holds for average-radius list decodability. Next, we show that such a result indeed holds (though with slightly weaker parameters). Indeed the result follows from the following simple observation:

\begin{proposition}
\label{prop:max-to-avg}
Let $\cC$ be a $(\rho,L)$-list decodable code. Then for any $\gamma>0$, $\cC$ is also $\left(\rho-\gamma,\frac{L}{\gamma}\right)$-average-radius list decodable.
\end{proposition}
\begin{proof}
Define $L'=L/\gamma$ and fix an arbitrary $\Lambda \subset \cC$ such that $|\Lambda|=L'$. Define
\[\Lambda_- = \Lambda\cap B_q(y,\rho n) \text{ and } \Lambda_+=\Lambda\setminus \Lambda_-.\]
Note that since $\cC$ is $(\rho,L)$-list decodable, we have $|\Lambda_-|\le L$. This implies that
\begin{equation}
\label{eq:Lambda-}
\sum_{c\in\Lambda_-} \agr(c,y) \le |\Lambda_-|\cdot n \le nL \le \gamma nL',
\end{equation}
where the last inequality follows from the definition of $L'$. Further, by the definition of $\Lambda_+$, we have
\[\sum_{c\in\Lambda_+} \agr(c,y) < (1-\rho)n\cdot |\Lambda_+| \le (1-\rho)nL'.\]
Combining the above with \eqref{eq:Lambda-} implies that $\sum_{c\in\Lambda} \agr(c,y) <(1-\rho+\gamma)nL'$, which completes the proof.
\end{proof}

Since a $(\rho,L_0)$-average-radius list decodable code is also $(\rho,L)$-list decodable, Propositions~\ref{prop:random-subcode} and~\ref{prop:max-to-avg} implies the following:
\begin{cor}
\label{cor:random-subcode-avg}
Let $\cC_0$ be an $(\rho,L_0)$-average-radius list decodable $q$-ary code. If we retain each codeword with probability $\frac{1}{q^{\eps n}\cdot L_0}$, then the resulting code with high probability is $(\rho-\eps,O(1/\eps^2))$-average-radius list decodable.
\end{cor}

\subsection{Lower Bound}

It is natural to wonder if one can pick a larger value of $p$ in Proposition~\ref{prop:random-subcode}, and whether the dependence $q^{\eps n}$ is necessary. In particular, if $L_0$ is only polynomial in $n$, could we pick $p=q^{-o(\eps n)}$? We will now argue that this is not possible.

First, we give a short argument, conditional on Conjecture~\ref{conj:listsize}.
By standard random coding argument, there exists a $q$-ary code $\cC_1$ with rate $1-H_q(\rho)-\frac{1}{n}$ that is $(\rho,n)$-list decodable. 
Suppose Proposition~\ref{prop:random-subcode} holds with $p=q^{-o(\eps n)}$.  If we applied this to the code $\cC_1$, we would obtain a code that is $(\rho,O(1/\eps))$-list decodable that has rate at least
\[1-H_q(\rho)-o(\eps),\]
assuming that $\eps$ is constant and $n$ is growing. However, this contradicts Conjecture~\ref{conj:listsize} for $L=O(1/\eps)$. 

Next, we argue an unconditional upper bound on $p$ in Proposition~\ref{prop:random-subcode}. In fact, we will prove something stronger: we will show that one needs $p=2^{-\Omega(\eps n)}$ even if the the original code $\cC_0$ has the stronger property of being $(\rho,L_0)$-average-radius list decodable (and the random subcode can have a weaker list decoding radius).

\begin{theorem}[Theorem \ref{thm:subcode-lb}, repeated]
For every $\rho > 0$, and for every $0 < \alpha < \frac{1 - \rho}{12}$, and for every $n$ sufficiently large, there exists a  code $\cC_0$ with block length $n$ that is $(\rho,n)$-average-radius list decodable such that the following holds. Let $\cC$ be  obtained by picking a random sub-code of $\cC_0$ of size $N = pN_0$ where  $p=q^{-\alpha n}/n$. Then with high probability if $\cC$ is $(\rho',L)$-list decodable for any $\rho'\ge 1/n$, then $L\ge \Omega(1/\alpha)$.
\end{theorem}

In the rest of the subsection, we will prove Theorem \ref{thm:subcode-lb}.

\subsubsection{Preliminaries}

We will need the following technical result, which follows the standard random coding argument and its analysis to determine the list decodability of random codes.

\begin{lemma}
\label{lem:random-centers}
Let $q\ge 2^{1/r}$ be an integer. Then there exists a code $\cC^*$ with rate $r$ and block length $n$ such that for every $2r<\gamma\le 1$, where $\gamma$ is a power of $1/2$, $\cC^*$ is $\left(1-\gamma, \left\lceil \frac{1}{\gamma-2r}\right\rceil\right)$-list decodable. Further, $\cC^*$ has relative distance $1-O(r)$.
\end{lemma}
\begin{proof} Fix a $\gamma$ with conditions as in the lemma statement. Let $\cC^*$ be a random code of rate $r$; by standard arguments, this distance of this code is $1 - O(r)$ with high probability~\cite{book}.  
 Further,  the standard random coding argument (see, for example,~\cite{book}) implies that $\cC^*$ is $(1-\gamma,L)$ list decodable except with probability at most
\[q^n\cdot q^{r n(L+1)}\cdot \left( \frac{q^{H_q(1-\gamma)n}}{q^n}\right)^{L+1}.\]
Rearranging, we can bound the expression above by
\begin{align}
& q^{-n(L+1)\left(1-H_q(1-\gamma)-r-\frac{1}{L+1}\right)}\notag\\
\label{eq:large-q}
&\qquad \le  q^{-n(L+1)\left(1-(1-\gamma+r)-r-\frac{1}{L+1}\right)}\\
 &\qquad = q^{-n(L+1)\left(\gamma-2r-\frac{1}{L+1}\right)} \notag\\
\label{eq:final-prob}
 &\qquad \le q^{-\Omega(n/L)}
\end{align}
In the above, \eqref{eq:large-q} follows from the following sequence of relations (that holds for any $0\le \rho \le 1-1/q$):
\[H_q(\rho)= \rho\log_q(q-1) + \frac{H_2(\rho)}{\log{q}} \le \rho +r,\]
where the inequality uses the fact that $q\ge 2^{1/r}$. \eqref{eq:final-prob} uses the fact that the choice of $L=\left\lceil \frac{1}{\gamma-2r}\right\rceil$ implies that $\gamma-2r-\frac{1}{L+1}>0$.

Finally, since the bound in \eqref{eq:final-prob} holds for any fixed $\gamma$ and that there are $O(\log(1/r))$ possible values of $\gamma$, the probability that the randomly chosen $\cC^*$ does not have the required property is $o(1)$, which completes the proof.
\end{proof}

\subsubsection{The Construction}

We now present the code $\cC_0$. Choose $\beta>0$ to be the smallest number such  that $1-\rho-\beta$ is a power of $1/2$. We will construct $\cC_0$ from $\cC^*$ as given by Lemma~\ref{lem:random-centers} with rate $r=(1-\rho-\beta)/6$.  (Note that by our choice of $\beta$, this implies that $r\ge (1-\rho)/12$ and hence, $\alpha<r$.)
 The construction goes as follows.
For every $c\in \cC^*$, let $N(c)$ be any $\frac{\beta\cdot n}{8\log(1/(1-\rho-\beta))}-1$ distinct vectors with Hamming distance $1$ from $c$. Then define
\[\cC_0=\cup_{c\in \cC^*} N(c).\]
Having constructed $\cC_0$, we argue next that it has good average-radius list-decodability. 
\begin{lemma}
\label{lem:c0-av-ld}
$\cC_0$ is $(\rho,n)$-average-radius list decodable.
\end{lemma}
\begin{proof} Recall that $1-\rho-\beta$ is a power of $1/2$.  
Fix an arbitrary $z$ and $\Lambda\subset \cC_0$ with $|\Lambda|=n$. We want to show that 
\begin{equation}
\label{eq:av-ld-bound}
\sum_{c\in \Lambda} \agr(z,c)< (1-\rho)n^2.
\end{equation}

Define
\[\B=B_q(z,(\rho+\beta)n).\]
We will break up the left-hand-side of \eqref{eq:av-ld-bound} into two parts, and handle $\Lambda \setminus \B$ and $\Lambda \cap \B$ separately.
First, we have
\begin{equation}
\label{eq:large-ball-bound}
\sum_{c\in\Lambda\setminus \B} \agr(z,c)< (1-\rho-\beta) n\cdot |\Lambda|= (1-\rho-\beta) n^2.
\end{equation}
Next, we bound $\sum_{c\in\Lambda\cap \B} \agr(z,c)$.  We break this sum up even further, and decompose $\B$ into the annuli 
\[\A_i:= B_q(z,(1-2^{-i-1})n)\setminus B_q(z,(1-2^{-i})n) \] 
for
$0\le i< \log\left(\frac{1}{1-\rho-\beta}\right)$.
Fix an $0\le i< \log\left(\frac{1}{1-\rho-\beta}\right)$ and for notational convenience define $\gamma=2^{-i-1}$.  (This will agree with the use of $\gamma$ in the statement of Lemma \ref{lem:random-centers}). 
Now consider $\Lambda \cap \A_i$, and consider the set 
\[ S := \inset{ c \in \cC^* \suchthat  N(c) \cap (\Lambda \cap \A_i) \neq \emptyset } \]
of ``centers" in $\cC^*$ whose ``clusters" $N(c)$ appear in this set.  We make the following two observations:
\begin{claim}
\label{clm:size-S}
$S \subset \cC^* \cap B_q(z, (1 - \gamma/2)n).$
\end{claim}
\begin{proof}
Since all vectors in $N(c)$ are at Hamming distance $1$ from a $c\in \cC^*$ (and $n$ is assumed to be large enough), we have that $c\in S$ implies that $c\in \A_{i-1}\cup\A_i\cup\A_{i+1}$. It is easy to see that the union of the three annuli is contained in $B_q(z, (1 - \gamma/2)n)$, which completes the proof.
\end{proof}

The following follows from the construction:
\begin{claim}
\label{clm:lambda-Ai-bound}
\[|\Lambda \cap \A_i| \le |S|\cdot \frac{\beta\cdot n}{8\log(1/(1-\rho-\beta))}.\]
\end{claim}

%
Thus, using the list-decodability of $\cC^*$ guaranteed by Lemma~\ref{lem:random-centers} and Claim~\ref{clm:size-S}, we have that $|S|\le \left\lceil \frac{1}{\gamma-2r}\right\rceil$. (Note that we can apply Lemma~\ref{lem:random-centers} since by our choice of parameters we have $\gamma\ge 1-\rho-\beta$, which in turn implies that $\gamma/2\ge (1-\rho-\beta)/2= 3r>2r$ as required.) Further, this with Claim~\ref{clm:lambda-Ai-bound} implies that
\begin{equation}
\label{eq:L-cap-Ai}
\left| \Lambda\cap\A_i\right| \le \left\lceil \frac{1}{\gamma-2r}\right\rceil \cdot \frac{\beta\cdot n}{8\log(1/(1-\rho-\beta))}\le \left( \frac{2}{\gamma-2r}\right) \cdot \frac{\beta\cdot n}{8\log(1-\rho-\beta)}.
\end{equation}
Now, we may bound
\begin{align}
\label{eq:step1}
\sum_{c\in\Lambda\cap\A_i} \agr(z,c)& \le \left| \Lambda\cap\A_i\right| \cdot 2\gamma n\\
\label{eq:step2}
& \le \left( \frac{1}{\gamma-2r}\right) \cdot \frac{\beta\cdot n}{4\log(1/(1-\rho-\beta))} \cdot 2\gamma n\\
& = \left( \frac{\gamma}{\gamma-2r}\right) \cdot \frac{\beta\cdot n^2}{2\log(1/(1-\rho-\beta))} \notag\\
\label{eq:stp3}
&\le \frac{\beta n^2}{\log(1/(1-\rho-\beta))}.
\end{align}
In the above, \eqref{eq:step1} follows from the fact that $\A_i$ lies outside of $B_q(z,(1-2\gamma)n)$. \eqref{eq:step2} follows from \eqref{eq:L-cap-Ai} while \eqref{eq:stp3} follows from the fact that $2r \le (1-\rho-\beta)/2\le \gamma/2$.
Finally, summing everything up and using \eqref{eq:large-ball-bound} and \eqref{eq:stp3}, we bound
\[ \sum_{c \in \Lambda \cap \B} \agr(c,z) \leq (1 - \rho - \beta)n^2 + \sum_{i=1}^{\log(1 - \rho -\beta)} \frac{ \beta n^2 }{ -\log(1 - \rho -\beta) } = (1 - \rho)n^2. \]
This establishes \eqref{eq:av-ld-bound}.
\end{proof}

\paragraph{Random subcodes of $\cC_0$ are typically not list-decodable.} 
Fix $\alpha > 0$, and 
let $\cC$ be a random subcode of $\cC_0$ of size $pN_0$, for $p = q^{-\alpha n}/n$ as in the statement of the theorem.
We finally argue that $\cC_0$ has many sub-codes that have terrible list decodability, thus proving Theorem~\ref{thm:subcode-lb}.
For any $z\in \cC^*$ and let $\Lambda(z)$ be an arbitrary subset of $N(z)$ such that $|\Lambda(z)|=D/\alpha$, where we will fix $D$ later. Further, order the ``centers" in $\cC^*$ as $z_1,z_2,\dots$. Then the following is the main technical lemma:
\begin{claim}\label{claim:prob-lb-alt}
For any $k\le \frac{q^{rn}}{3}$, we have
\[\PR{\Lambda(z_{k+1})\subset \cC| \cC\cap \left(\cup_{i=1}^k N(z_i)\right)} \ge
\left(\frac{p}{2e}\right)^{\frac{D}{\alpha}} \geq q^{-2Dn}.\]
\end{claim}

Once we establish Claim~\ref{claim:prob-lb-alt}, we are done.  Indeed, we have
\begin{align*}
\PR{ \forall i \leq \frac{q^{rn}}{3} , \Lambda(z_i) \not\subset \cC }
&= \prod_{k=0}^{q^{rn}/3} \PR{ \Lambda(z_{k+1}) \not\subset \cC \mid \cC \cap \inparen{ \cup_{i=1}^k \Lambda(z_i) } } \\
&\leq \inparen{1 - q^{-2Dn} }^{q^{rn}/3}.
\end{align*}
Thus, the probability that there is some $i$ with $\Lambda(z_i) \subset \cC$ is at least
\begin{equation}
\label{eq:low-prob}
1-(1-q^{-2Dn})^{q^{r n}/3}\ge 1-e^{-q^{(r-2D)n}/3} \ge 1-o(1),
\end{equation}
where the last inequality follows if we pick $D=r/3$. 
In particular, $\cC$ has list sizes at least $|\Lambda(z_i)|$, even at distance $\rho = 1/n$, which is the radius of $\Lambda(z_i)$.

We conclude by proving Claim~\ref{claim:prob-lb-alt}. 
For notational convenience, define $\cC_0^{(k)}=\cup_{i=1}^k N(z_i)$; thus, $\cC_0^{(k)}$ is the code $\cC_0$ after the first $k$ clusters $N(z_i)$ have been added.  
Let $M_k=|\cC\cap \cC_0^{(k)}|$. Note that $M_k$ is random variable. 
For $k > 0$, let $N_k = |\cC_0^{(k)}| = k\cdot |N(z_1)|=\frac{\beta\cdot nk}{8\log(1/(1-\rho-\beta))}$.  (The fact that for $k>0$, $N_k = k|N(z_1)|$ follows because the sets $N(z_i)$ are all disjoint, which itself follows from the distance of the code and the fact that all the clusters are of the same size).

The main observation is that conditioned on $\cC\cap \cC_0^{(k)}$, the distribution on $\cC$ is the same as the distribution where $pN_0 - M_k$ codewords are picked uniformly at random, with replacement, from $\cC_0\setminus \cC_0^{(k)}$.  
Again, this follows because the clusters $N(z_i)$ are disjoint.  Call this distribution $\mu$.
%
For all $k$, we have $M_k < M_{q^{rn}/3}$.  A Chernoff bound implies that this latter is small: 
\[ \PR{ M_{q^{rn}/3} \geq pN_0/2 } \leq \exp\inparen{ -\Omega(pN_0) }. \]
We will absorb this failure probability into the calculation in \eqref{eq:low-prob}, and assume from now on that $M_k < pN_0/2$ for all $k$.
Now, the probability we need to bound to prove Claim~\ref{claim:prob-lb-alt} is 
\[ \PR{ \Lambda(z_{k+1}) \subset \cC | \cC \cap \cC^{(k)} } = \mathbb{P}_{\mu} \inbrac{ \forall v \in \Lambda(z_{k+1}), v \in \cC } 
\geq \mathbb{P}_\mu \inbrac{ \forall v \in \Lambda(z_{k+1}), v \in \cC \text{ exactly once } }. \]
Let $D' = D/\alpha$.  Then the right hand side above, 
the probability that each of the vectors in $\Lambda(z_{k+1})$ is picked exactly once in $\cC$ under $\mu$, is given by
\[\binom{pN_0-M_k}{D'}\left(1-\frac{D'}{N_0-N_k}\right)^{pN_0-M_k-D'} \frac{(D')!}{(N_0-N_k)^{D'}} \ge
\binom{pN_0/2}{D'}\left(1-\frac{2D'}{N_0}\right)^{pN_0-D'} \frac{(D')!}{(N_0/2)^{D'}},\]
where the inequality follows from the fact that $pN_0\ge pN_0-M_k>pN_0/2$ and that for all $k$,  
\[N_0-N_k \geq N_0 - N_{q^{rn}/3} \ge 2N_0/3>N_0/2, \]
where the second inequality follows from the fact that all the clusters $N(z_i)$ have the same size.
Now it suffices to bound the last expression from below by $\left(\frac{p}{2e}\right)^{D'}$.  And indeed, we have
\begin{align*}
\binom{pN_0/2}{D'}\left(1-\frac{D'}{N_0/2}\right)^{pN_0-D'} \frac{(D')!}{(N_0/2)^{D'}} & \ge \left(\frac{pN_0}{2D'}\right)^{D'}\cdot \left(1-\frac{2D'}{N_0}\right)^{pN_0}\cdot \left(\frac{2D'}{eN_0}\right)^{D'}\\
& = \left(\frac{p}{e}\right)^{D'}\cdot \left(\left(1-\frac{2D'}{N_0}\right)^{N_0/(2D')}\right)^{2pD'}\\
&\ge \left(\frac{p}{4^{2p}e}\right)^{D'}\\
&\ge \left(\frac{p}{2e}\right)^{D'}.
\end{align*}
In the above the second inequality follows for $N_0\ge 4D'$ and the final inequality follows for $p\le 1/4$ both of which are valid assumptions for our choices for $p$ and $N_0$.
Finally, we have
\[\left(\frac{p}{2e}\right)^{\frac{D}{\alpha}} \ge \left(\frac{1}{2enq^{\alpha n}}\right)^{\frac{D}{\alpha}}\ge \left(\frac{1}{q^{2\alpha n}}\right)^{\frac{D}{\alpha}}=q^{-2Dn},\]
for large enough $n$, which completes the claim.

\end{document}